%% file: main.tex
\documentclass[sigconf, nonacm]{acmart}

\input{preamble}

\begin{document}
\title{\system{}: Certified Reuse of Cached Natural-Language Views for Semantic Operators}

\input{authors}

\input{abstract}

\maketitle

\vldbtopmatter

\input{introduction}
\input{preliminaries}
\input{methodology}

\input{experiment}
\input{relatedwork}
\input{conclusion}

\clearpage
\newpage
\bibliographystyle{ACM-Reference-Format}
\bibliography{refs}

\end{document}

%% file: preamble.tex
\usepackage{pvldb}

\renewcommand\vldbdoi{XX.XX/XXX.XX}
\renewcommand\vldbpages{XXX-XXX}
\renewcommand\vldbavailabilityurl{https://github.com/LeoTTTPhat/semreuse}

\usepackage{booktabs}
\usepackage{graphicx}
\usepackage[ruled,vlined,linesnumbered]{algorithm2e}
\usepackage{tikz}
\usepackage{pgfplots}
\usepackage{caption}
\usepackage{subcaption}
\usetikzlibrary{arrows.meta, positioning, shapes.geometric}
\usepackage{pgfplots}
\pgfplotsset{compat=1.18}
\usepgfplotslibrary{groupplots}
\input{figures/pgfstyle}

\newtheorem{proposition}{Proposition}

\newtheorem{definition}{Definition}

\newcommand{\system}{\texttt{SemReuse}\xspace}
\usepackage{xspace}
\newcommand{\Ora}{\mathcal{O}}
\newcommand{\ext}{\mathrm{ext}}
\newcommand{\rel}[1]{\textsc{#1}}

%% file: figures/pgfstyle.tex
\definecolor{srBlue}{HTML}{2A78D6}
\definecolor{srOrange}{HTML}{EB6834}
\definecolor{srAqua}{HTML}{1BAF7A}
\definecolor{srTa}{HTML}{FFD966}
\definecolor{srTb}{HTML}{CE7E00}
\definecolor{srTc}{HTML}{256ABF}
\definecolor{srTd}{HTML}{104281}
\definecolor{srInk}{HTML}{0B0B0B}
\definecolor{srInkTwo}{HTML}{52514E}
\definecolor{srMuted}{HTML}{898781}
\definecolor{srGrid}{HTML}{E1E0D9}
\definecolor{srAxis}{HTML}{C3C2B7}

\pgfplotsset{
  sr base/.style={
    font=\footnotesize,
    scale only axis,
    axis x line*=bottom, axis y line*=left,
    axis line style={draw=srAxis, line width=0.4pt},
    every tick/.style={draw=srAxis, line width=0.4pt},
    major tick length=2pt, minor tick length=1.2pt,
    tick label style={font=\scriptsize, text=srInk},
    label style={font=\footnotesize, text=srInk},
    title style={font=\footnotesize, text=srInk, align=left},
    grid style={draw=srGrid, line width=0.3pt},
    ymajorgrids,
    legend style={draw=none, fill=none, font=\scriptsize, text=srInk,
                  /tikz/every even column/.append style={column sep=6pt}},
    legend cell align=left,
    every axis plot/.append style={line width=0.9pt, line join=round,
                                   line cap=round},
  },
  sr mark/.style={mark options={solid, draw=white, line width=0.5pt,
                                scale=1.0}},
  sr open mark/.style={mark options={solid, fill=white, line width=0.7pt}},
}

%% file: authors.tex
\author{Phat T. Tran-Truong}
\affiliation{%
  \institution{Vietnam National University}
  \city{Ho Chi Minh City}
  \country{Vietnam}
}
\email{phatttt@hcmut.edu.vn}

\author{Tung Kieu}
\affiliation{%
  \institution{Aalborg University}
  \city{Aalborg}
  \country{Denmark}
}
\email{tungkvt@cs.aau.dk}

\author{Xuan-Bach Le}
\affiliation{%
  \institution{Vietnam National University}
  \city{Ho Chi Minh City}
  \country{Vietnam}
}
\email{lexuanbach@hcmut.edu.vn}

\author{Son X. Ha}
\affiliation{%
  \institution{RMIT University}
  \city{Ho Chi Minh City}
  \country{Vietnam}
}
\email{ha.son@rmit.edu.vn}

\author{Duc Kieu}
\affiliation{%
  \institution{Deakin University}
  \city{Geelong}
  \country{Australia}
}
\email{v.kieu@deakin.edu.au}

\author{Bin Yang}
\affiliation{%
  \institution{East China Normal University}
  \city{Shanghai}
  \country{China}
}
\email{byang@dase.ecnu.edu.cn}

\author{Christian S. Jensen}
\affiliation{%
  \institution{Aalborg University}
  \city{Aalborg}
  \country{Denmark}
}
\email{csj@cs.aau.dk}

%% file: abstract.tex
\begin{abstract} 
Semantic operators allow users to filter unstructured data using conditions expressed in natural language, but existing systems typically invoke an LLM for every input tuple. Analytical workloads often contain repeated or related filters whose results could be reused. Such reuse is risky because it requires inferring whether filter conditions are equivalent, whether one implies another, or whether they are mutually exclusive. An incorrect judgment can omit qualifying tuples or include nonqualifying ones. We introduce \system{}, a system for safely reusing previous semantic-filter results. \system{} stores each result as a materialized view and compares new filters with existing views to identify tuples that can provisionally be accepted or skipped without another LLM call. Because these decisions rely on a fallible language reasoner, \system{} audits every rewritten result. Its pooled stratified audit uses exact finite-population bounds to provide jointly valid per-query recall and precision guarantees. If the requested guarantees cannot be established, \system{} evaluates additional tuples and progressively approaches full evaluation. Thus, incorrect relation judgments affect efficiency rather than the validity of the reported guarantees. Experiments on large labeled corpora and real LLM executions show reductions of up to $15.4\times$. They also identify cross-filter inconsistency as the main source of repair cost. 
\end{abstract}

%% file: introduction.tex
\section{Introduction}
\label{sec:introduction}

Semantic operators bring a relational interface to analytics over unstructured data~\cite{patel2025lotus,liu2025palimpzest,shankar2025docetl,dai2024uqe,liskowski2025aisql}. A semantic filter asks a large language model (LLM) whether each input tuple satisfies a condition expressed in natural language. For example,
\begin{center}
\texttt{sem\_filter(``the document is about baseball'')}
\end{center}
evaluates the condition independently for every tuple. Filtering a relation containing one million documents can therefore require one million LLM calls. Existing systems reduce this cost within an individual query through model cascades and proxies~\cite{kang2020supg,lu2018probabilistic,zeighami2026bargain}, plan optimization~\cite{russo2025abacus,zhao2026larch}, specialized indexes~\cite{lin2024zendb,sun2025quest}, and approximate processing~\cite{jo2024thalamusdb}. These techniques largely optimize each query independently, although analytical workloads often contain related filters. An analyst may first search for documents about ``sports'' and later narrow the analysis to documents about ``baseball.'' In this case, tuples previously rejected by the broader filter need not be evaluated again. Similarly, a repeated or rephrased condition may reuse nearly the entire earlier result. Relational database systems have long exploited such opportunities by storing previous query results as materialized views and rewriting later queries to use them, a problem traditionally studied as \emph{answering queries using views}~\cite{levy1995answering,halevy2001aquv,chaudhuri1995optimizing,goldstein2001optimizing}.

\noindent\textbf{The challenge.} Reusing a previous semantic-filter result requires determining how its filter condition relates to the new one. For SQL conditions, an optimizer can often establish such relationships exactly. Conditions expressed in natural language do not provide an exact algebra. A language model may judge that one condition implies another, that the two conditions are equivalent, or that they are mutually exclusive, but any of these judgments can be wrong. An incorrect implication may exclude qualifying tuples, while an incorrect equivalence may introduce nonqualifying tuples.

Existing semantic caches do not resolve this problem. They typically reuse an entire cached response when two requests have similar embeddings. However, similarity is symmetric, whereas implication is directional: about ``baseball'' may imply about ``sports,'' but not conversely. Moreover, semantic caches either reuse the complete result or do not reuse it at all, whereas semantic filters can often benefit from reusing only part of a previous result. Most importantly, these caches do not bound the errors introduced by an incorrect reuse decision. This raises the central question of this work: \emph{How can a semantic-operator system reuse results across related filter queries while providing reliable guarantees even when the relationships between their filter conditions are judged incorrectly?}

\noindent\textbf{The \system{} approach.} We present \system, a system for certified reuse of previously executed semantic filters. \system records each filter result as a materialized view containing the filter condition, a bitmap of the reported tuples, verification information, and lineage. When a new filter arrives, a lightweight two-tier reasoner compares its condition with those of the cached views. The reasoner proposes whether two conditions are equivalent, whether one implies the other, whether they are mutually exclusive, or whether no useful relationship can be established.

The rewriter translates these proposed relationships into tuple-level reuse. A broader cached view can restrict the tuples that require evaluation; a narrower view can provide tuples that are provisionally accepted; and a mutually exclusive view can identify tuples that are provisionally rejected. Multiple views can be combined through bitmap operations to construct a reduced evaluation plan. These decisions can substantially decrease the number of tuples sent to the LLM, but they are treated as hypotheses rather than trusted facts.

\system\ therefore audits every rewritten result. The plan partitions the relation into candidate tuples that are evaluated by the LLM, provisionally accepted tuples, and provisionally rejected tuples. The audit checks samples from the provisionally accepted tuples and from a pooled population of the rejected tuples. Exact finite-population bounds then provide jointly valid lower bounds on recall and precision for the current query. Errors discovered during the audit are corrected. If the recall bound does not meet the requested target, \system\ evaluates additional groups of tuples until the target is reached. In the worst case, it approaches full evaluation. Consequently, an incorrect relationship proposed by the reasoner can increase execution cost, but it does not invalidate the reported guarantees.

This separation between planning and verification changes how the reasoner should be used. A conventional optimizer acts only on highly confident predictions because an incorrect decision directly damages the result. In \system, the audit bounds and repairs such errors. The reasoner can therefore propose reuse more aggressively, while the audit determines how much of that reuse can safely be retained. Our experiments confirm that lower confidence thresholds can produce greater savings without changing the validity of the certificates---a principle we summarize as \emph{trust aggressively, audit everything}.

We evaluate \system using large single-label and multi-label corpora, real LLM evaluators from four model families over an analyst-written sequence of filters, a tuned per-query proxy cascade, and the released LOTUS engine; we also ablate each audit design decision and test the certificates where they are tight. On the large labeled workloads, \system reduces LLM calls by $2.4\times$--$15.4\times$ while providing per-query certificates. The real-LLM experiments reveal an important limitation, on every model: a relationship that is reasonable at the level of language may not hold consistently across the LLM's tuple-level decisions. In this case, the audit remains valid, but repairing the resulting plan can consume much of the reuse benefit. Thus, the effectiveness of cross-query reuse depends not only on corpus size or reasoner quality, but also on the consistency of the LLM evaluator across related filter conditions.

In summary, we make four contributions: (1) We formulate cross-query reuse for semantic filters and define reuse rules based on equivalence, implication, and mutual exclusion between filter conditions. (2) We present \system, which combines a materialized-view store, a two-tier relation reasoner, and a bitmap-based rewriter to reduce tuple-level LLM evaluations. (3) We develop a pooled stratified audit that provides jointly valid per-query recall and precision certificates, regardless of whether the inferred relationships are correct. (4) We evaluate \system across diverse corpora, workloads, and LLM settings, demonstrating substantial savings and identifying cross-filter consistency as the main practical constraint.

%% file: preliminaries.tex
\section{Background and Problem Formulation}

\begin{definition}[Semantic Filter]
    Let $D=\{x_1,\ldots,x_N\}$ be a relation of unstructured tuples, and let $\mathcal{P}$ denote the space of filter conditions expressed in natural language. An \emph{oracle} is a function
    \begin{equation}
    \Ora:\mathcal{P}\times D\rightarrow{0,1},
    \end{equation}
    where $\Ora(p,x)=1$ indicates that tuple $x$ satisfies condition $p$. The result of fully evaluating $p$ over $D$ is its \emph{extension}
    \begin{equation}
    \ext(p)={x\in D:\Ora(p,x)=1}.
    \end{equation}
    In deployment, $\Ora$ is typically an LLM, and each evaluation $\Ora(p,x)$ incurs one oracle call.
\end{definition}

We assume that $\Ora(p,x)$ is fixed for every condition--tuple pair, for example by fixing the model, prompt, and decoding procedure or by memoizing its output. This assumption makes $\ext(p)$ defined well. The oracle defines the reference semantics of the semantic filter and is not assumed to coincide with human ground truth.

\begin{definition}[Cached View]
    Executing a filter condition $q\in\mathcal{P}$ produces a cached view
    \begin{equation}
    V_q=(q,R_q,\mathit{ver}_q,\mathit{lineage}_q),
    \label{eq:cached-view}
    \end{equation}
    where $R_q\subseteq D$ is the reported result, stored as a bitmap; $\mathit{ver}_q\subseteq D$ is the set of tuples evaluated directly by the oracle; and $\mathit{lineage}_q$ records the data and oracle versions, source views, and published certificates. A view is \emph{exact} if $R_q=\ext(q)$. Otherwise, it may contain false inclusions or exclusions introduced by earlier reuse. Let $\mathcal{V}$ denote the set of cached views available before the current query.
\end{definition}

\begin{definition}[Extension Relationship]
    For two filter conditions $p,q\in\mathcal{P}$, let $P=\ext(p)$ and $Q=\ext(q)$. Their relationship is defined by
\begin{equation}
\mathsf{Rel}(p,q)=
\begin{cases}
\rel{equiv},    & P=Q,\\
\rel{forward},  & \emptyset\neq P\subsetneq Q,\\
\rel{backward}, & \emptyset\neq Q\subsetneq P,\\
\rel{disjoint}, & P\cap Q=\emptyset \text{ and } P\neq Q,\\
\rel{overlap},  & \text{otherwise}.
\end{cases}
\label{eq:relations}
\end{equation}
    The relation names are defined with respect to the ordered pair $(p,q)$. Under \rel{forward}, $q$ is broader than $p$; under \rel{backward}, $q$ is narrower than $p$.
    
    Computing $\mathsf{Rel}(p,q)$ exactly requires the complete extensions of both conditions. Instead, \system\ employs a relation reasoner
    \begin{equation}
    E:\mathcal{P}\times\mathcal{P}
    \rightarrow \mathcal{R}\times[0,1],
    \qquad
    E(p,q)=(\widehat{\mathit{rel}},c),
    \end{equation}
    where $\mathcal{R}$ contains the five relationships in~\eqref{eq:relations}, $\widehat{\mathit{rel}}$ is the predicted relationship, and $c$ is its confidence. Because $E$ observes only the filter conditions, $\widehat{\mathit{rel}}$ may differ from the true extension relationship. The reasoner's output is therefore used to construct an execution plan, not as a correctness guarantee.
\end{definition}

\textbf{Certified reuse problem.}
Let $\mathcal{H}$ denote the execution history available before the current audit, including the cached views and preceding query results. Given a new filter condition $p$, a view store $\mathcal{V}$, a relation reasoner $E$, a recall target $t\in(0,1]$, and a confidence level $1-\alpha$, the objective is to produce a reported result $R_p\subseteq D$ and lower bounds $b_R,b_P\in[0,1]$ such that
\begin{equation}
\Pr_{\mathcal{A}}\!\left[
\mathrm{recall}_p(R_p)\ge b_R
;\land\;
\mathrm{precision}_p(R_p)\ge b_P
,\middle|\,\mathcal{H}
\right]
\ge 1-\alpha,
\;
b_R\ge t,
\label{eq:certificate}
\end{equation}
where $\mathcal{A}$ denotes the fresh randomness used to audit the current query, and
\begin{equation}
\mathrm{recall}_p(R_p)
=\frac{|R_p\cap\ext(p)|}{|\ext(p)|},
\;\;
\mathrm{precision}_p(R_p)
=\frac{|R_p\cap\ext(p)|}{|R_p|}.
\end{equation}
A ratio is defined as one when its denominator is zero.

Let $C_{\Ora}(p)$ denote the number of oracle calls issued while answering $p$. The optimization problem is
\begin{equation}
\min\mathbb{E}_{\mathcal{A}}
\!\left[C_{\Ora}(p)\mid\mathcal{H}\right]
\;\;
\text{subject to~\eqref{eq:certificate}.}
\end{equation}
The guarantee is per query and is defined relative to the oracle result $\ext(p)$. It makes no assumption about the accuracy of $E$ or the correctness of the cached views.

%% file: methodology.tex
\section{\system: System Design}
\label{sec:system}

\subsection{Framework Overview}

\input{figures/fig_framework}

Figure~\ref{fig:framework} summarizes the execution of a semantic filter. For a new predicate $p$, \system\ first retrieves eligible cached views, infers the relationship between $p$ and each cached predicate, and constructs a reuse plan. It may then refine the plan using an optional per-query proxy. Finally, it evaluates the remaining candidate tuples and audits all reused decisions. View lookup, relation inference, and plan construction do not invoke the tuple-level oracle $\Ora(p,x)$; candidate evaluation, proxy pilots, and auditing account for the oracle-call budget analyzed in this paper.

\subsection{Predicate Store}

The predicate store maintains one \texttt{CachedView} for each executed predicate, using the representation defined in Equation~\ref{eq:cached-view}. A hash index over normalized predicate text identifies exact repetitions without invoking the relation reasoner. Such a match establishes predicate equivalence, but its stored result is still audited unless the view is exact. For each view $V_{q}$, let $D_{q}\subseteq D$ denote the tuples covered when the view was created, as determined from its lineage. This coverage information ensures that tuples appended after the view was created are not treated as cached negatives.

The result $R_{q}$ and verification set $\mathit{ver}_q$ are stored as bitmaps. Thus, excluding metadata, a view requires approximately two bits per covered tuple; for example, 100 views over 100,000 tuples require approximately 2.5MB. Each view also records its oracle model, prompt-template version, source views, and published certificates. Only views with compatible oracle and prompt versions are eligible for reuse. Our prototype uses an in-memory store, although the system interface does not depend on a particular persistence layer.

\subsection{Calibrated Two-Tier Relation Reasoner}

For a new predicate $p$ and a cached predicate $q$, the relation reasoner returns $E(p,q)=(\widehat{\mathit{rel}},c)$, where $\widehat{\mathit{rel}}$ is one of the five relationships defined in Equation~\ref{eq:relations}, and $c\in[0,1]$ is the confidence used during plan construction. The reasoner operates once per eligible cached view, so its cost depends on the number of views rather than the number of tuples.

\subsubsection{Tier one: bidirectional NLI scoring}
The first tier uses a DeBERTa-v3-xsmall cross-encoder fine-tuned on SNLI and MNLI datasets~\cite{he2023debertav3,bowman2015snli,williams2018multinli}. The model has 71 million backbone parameters and can be executed on a CPU.

A natural-language-inference model examines an ordered pair consisting of a \emph{premise} and a \emph{hypothesis}. It returns three normalized scores corresponding to \textit{entailment}, \textit{contradiction}, and \textit{neutrality}. We evaluate the predicate pair in both directions:
\begin{equation}
\begin{aligned}
\boldsymbol{\pi}_{p\rightarrow q}
=
\left(
\pi^{\mathrm{ent}}_{p\rightarrow q},
\pi^{\mathrm{con}}_{p\rightarrow q},
\pi^{\mathrm{neu}}_{p\rightarrow q}
\right),\;\;
\boldsymbol{\pi}_{q\rightarrow p}
=
\left(
\pi^{\mathrm{ent}}_{q\rightarrow p},
\pi^{\mathrm{con}}_{q\rightarrow p},
\pi^{\mathrm{neu}}_{q\rightarrow p}
\right)
\end{aligned}
\end{equation}
Each vector contains three nonnegative values that sum to one. The first vector measures whether the text of $p$ supports, contradicts, or is neutral with respect to $q$; the second asks the converse question. Concatenating the two three-dimensional score vectors yields a six-dimensional representation: $\mathbf{z}(p,q) = \left[ \boldsymbol{\pi}_{p\rightarrow q}; \boldsymbol{\pi}_{q\rightarrow p} \right] \in[0,1]^6$.

Bidirectional entailment provides evidence for \rel{equiv}; entailment only from $p$ to $q$ provides evidence for \rel{forward}; and entailment only from $q$ to $p$ provides evidence for \rel{backward}. Contradiction provides evidence for \rel{disjoint}, while the remaining patterns generally indicate \rel{overlap}. However, directly applying these rules is unreliable because an NLI model is trained to classify sentence pairs rather than extension relationships between semantic filters, and its raw scores are not necessarily well calibrated~\cite{guo2017calibration}.

We therefore pass $\mathbf{z}(p,q)$ to a five-class multinomial logistic-regression head: $\boldsymbol{\rho}(p,q) = \operatorname{softmax}\!\left( \mathbf{W}\mathbf{z}(p,q)+\mathbf{b} \right)$, where the five entries of $\boldsymbol{\rho}$ correspond to the relationships in Equation~\ref{eq:relations}. The tier-one prediction and confidence are $\widehat{\mathit{rel}}_{1} = \arg\max_{r\in\mathcal R}\rho_r(p,q)$ and $c_1 = \rho_{\widehat{\mathit{rel}}_1}(p,q)$. The head is trained on approximately 300 labeled predicate pairs generated from a synthetic taxonomy of invented topics. These pairs are independent of all evaluation corpora.

\subsubsection{Tier two: LLM arbitration}
Tier one may be uncertain or may predict \rel{overlap}, which does not enable a rewrite. We therefore invoke a second-tier arbiter when $c_1<\tau_2$ or $\widehat{\mathit{rel}}_1=\rel{overlap}$. The arbiter receives only the two predicate sentences and the definitions of the five relationships, and it is instructed to select one relationship. In our implementation, it uses the same instruction-tuned model deployed as the tuple-level oracle; however, an arbitration request contains no corpus tuple and is substantially smaller than a semantic-filter evaluation.

We do not ask the arbiter to verbalize its confidence because such self-reported values are generally not calibrated. Instead, if the arbiter predicts relation $r$, we assign it the empirical precision of class $r$ measured on the corpus-independent development taxonomy:
\begin{equation}
    c_2(r)
    =
    \frac{
        \#\{\text{development pairs correctly predicted as }r\}
    }{
        \#\{\text{development pairs predicted as }r\}
    }.
\end{equation}
The tier-two prediction replaces the tier-one result. Every ordered predicate pair is memoized together with the reasoner, model, and prompt versions, so the arbitration cost is paid at most once for each configuration.

Neither tier is assumed to be correct. The predicted relationship and its confidence determine which reuse plan is attempted. An incorrect prediction may increase auditing or escalation costs, but the audit in Section~\ref{sec:audit} establishes the final guarantee. Section~\ref{sec:eval-entailment} evaluates the classification accuracy, calibration, and cost of the two-tier reasoner.

\subsection{Reuse Rules and Plan Construction}
\label{sec:rewrite}

The rewriter receives the judgments $\{(V_q,\widehat{\mathit{rel}}_{p,q},c_{p,q}) : V_q\in\mathcal V\}$. A judgment is eligible if its confidence exceeds the threshold associated with its predicted relation: $c_{p,q}\ge \tau_{\widehat{\mathit{rel}}_{p,q}}$ and $\widehat{\mathit{rel}}_{p,q} \neq\rel{overlap}$. The current implementation greedily selects eligible views in decreasing confidence order, subject to a maximum of $k_{\max}$ views. This cap limits the number of audit strata. It affects efficiency but not correctness.

Each selected view produces one or more provisional tuple-level claims:

\begin{description}
\item[R1: Superset pruning.]
If $\widehat{\mathit{rel}}_{p,q}\in
\{\rel{forward},\rel{equiv}\}$, then $q$ is predicted to contain $p$. Tuples in $D_q\setminus R_q$ are therefore provisionally classified as negative for $p$.

\item[R2: Positive reuse.]
If $\widehat{\mathit{rel}}_{p,q}\in
\{\rel{backward},\rel{equiv}\}$, then $q$ is predicted to be contained in $p$. Tuples in $R_q$ are therefore provisionally classified as positive for $p$.

\item[R3: Disjoint elimination.]
If $\widehat{\mathit{rel}}_{p,q}=\rel{disjoint}$, tuples in $R_q$ are provisionally classified as negative for $p$.

\item[R4: Equivalence.]
An equivalence judgment applies both R1 and R2. It provisionally reuses $R_q$ as the complete result over the view's coverage $D_q$. Tuples in $D\setminus D_q$ remain candidates.
\end{description}

These claims can be composed formally. Let
$\mathcal V_{\mathrm{sup}}$,
$\mathcal V_{\mathrm{pos}}$, and
$\mathcal V_{\mathrm{dis}}$
denote the selected views assigned to R1, R2, and R3, respectively. Equivalence views belong to both
$\mathcal V_{\mathrm{sup}}$ and
$\mathcal V_{\mathrm{pos}}$.
The accumulated positive and negative claims are
\begin{equation}
A_p^+
=
\bigcup_{V_q\in\mathcal V_{\mathrm{pos}}} R_q,
\;\;
A_p^-
=
(
\bigcup_{V_q\in\mathcal V_{\mathrm{sup}}}
(D_q\setminus R_q)
)
\cup
(
\bigcup_{V_q\in\mathcal V_{\mathrm{dis}}}
R_q
).
\label{eq:pos-and-neg}
\end{equation}
A tuple may receive both a positive and a negative claim. Because neither claim can be trusted before auditing, all such conflicts are returned to oracle evaluation. The three plan components are therefore
\begin{equation}
P_p^+ = A_p^+\setminus A_p^-,
\;\;\;
P_p^- = A_p^-\setminus A_p^+,
\;\;\;
C_p = D\setminus(P_p^+\cup P_p^-),
\label{eq:plan-partition}
\end{equation}
where $P_p^+$ contains the provisionally accepted tuples, $P_p^-$ contains the provisionally rejected tuples, and $C_p$ is the candidate set evaluated by the oracle. These sets form a disjoint partition of $D$. With no eligible view, $P_p^+=P_p^-=\emptyset$ and $C_p=D$, recovering cold evaluation.

Before auditing, the rewritten result is
\begin{equation}
\widetilde R_p
=
P_p^+
\cup
{x\in C_p:\Ora(p,x)=1}.
\label{eq:rewritten-result}
\end{equation}
The rewriter partitions $P_p^+$ and $P_p^-$ into disjoint strata according to the source view, rewrite rule, and confidence. When several views make the same nonconflicting claim, Algorithm~\ref{alg:plan} applies a deterministic priority order so that each assumed tuple belongs to exactly one stratum. These strata provide the provenance needed by the audit for sampling, correction, and escalation.

\input{algoritms/buildplan}

The following theorem characterizes the ideal case in which the provisional claims are exact.

\begin{theorem}[Rewrite correctness under exact inputs]
\label{thm:c1}
Assume that the oracle is deterministic. For every view $V_q$ selected by Algorithm~\ref{alg:plan}, suppose that
\(
R_q=\ext(q)\cap D_q
\)
and that
\(
\widehat{\mathit{rel}}_{p,q}
=
\mathsf{Rel}(p,q).
\)
Then the result $\widetilde R_p$ in Equation~\ref{eq:rewritten-result} equals $\ext(p)$, for each individual rewrite rule and for any combination of the rules.
\end{theorem}

The exact-view condition is necessary: a correct relationship between $p$ and $q$ does not, by itself, make an approximate $R_q$ exact. In practice, \system\ does not assume either condition of Theorem~\ref{thm:c1}. The audit verifies the plan directly against $\Ora(p,\cdot)$. Consequently, the confidence thresholds control a cost tradeoff: lower thresholds create more reuse opportunities but may require more auditing and escalation.

\subsection{Optional Composition with a Proxy Cascade}
\label{sec:compose}

Cross-query reuse and per-query proxy optimization are complementary. Reuse excludes tuples using results from earlier predicates, whereas a proxy ranks tuples using inexpensive scores learned specifically for the current predicate. We therefore apply an optional proxy only to the candidate set $C_p$ remaining after reuse, rather than to the full relation.

The proxy stage draws a uniform pilot sample from $C_p$, obtains its oracle labels, and fits a lightweight scoring model. Separate pilot subsets are used to select the cutoff and determine whether the proxy should be activated. The proxy is discarded unless validation indicates that it retains sufficiently many positives while pruning enough tuples to justify its pilot cost.

When activated, the proxy retains tuples above the cutoff as candidates and partitions those below it into provisionally negative score bands. These proxy-derived strata are combined with the strata created by the reuse plan and certified by the same audit. To preserve validity, the complete partition is fixed before fresh audit samples are drawn. If the certified recall lower bound remains below the target, \system\ evaluates the proxy strata in descending order of score. Our evaluation compares reuse alone, the proxy cascade alone, and their combination while charging all pilot, candidate-evaluation, auditing, and escalation calls to the same oracle-call budget.

\subsection{System Integration, Updates, and Scope}
\label{sec:deploy}

\system\ interposes immediately before a semantic filter issues tuple-level oracle calls. The host supplies an oracle interface $\Ora(p,\cdot)$ and the relevant tuple identifiers. This interface may internally use batching, a cascade, or a remote endpoint, provided that it returns a fixed binary response for each predicate--tuple pair. \system\ returns the result bitmap and recall and precision certificates relative to the host oracle without modifying the host's existing optimizer. We validate this integration using the released LOTUS implementation and measure the reduction in its oracle calls.

Reuse is restricted to lineage-compatible views whose oracle model and prompt-template versions match the current query. A model or prompt update therefore makes earlier views ineligible. Newly appended tuples fall outside a view's coverage $D_q$ and enter the candidate set, while deleted tuples are removed by masking their bitmap positions.

Our claims and evaluation concern unary semantic filters. The same auditing principle could be applied to tuple pairs for semantic joins, although the larger and typically sparser pair space may increase certification cost. Semantic maps generally support only exact or equivalence-based reuse because they lack subset structure. We leave both operators to future work.

\section{Audit and Certification}
\label{sec:audit}

A rewrite plan avoids oracle calls by provisionally accepting or rejecting some tuples. Errors among the provisionally accepted tuples reduce precision, whereas qualifying tuples among the provisionally rejected tuples reduce recall. The audit samples both groups, corrects the errors it observes, and converts the remaining uncertainty into the certificate required by Equation~\ref{eq:certificate}.

Before auditing begins, the plan partitions the relation into the candidate set $C_p$, assumed-positive strata $A_1,\ldots,A_S$, and pruned strata $P_1,\ldots,P_J$:
\begin{equation}
D
=
C_p
\cup
\bigcup_{i=1}^{S} A_i
\cup
\bigcup_{j=1}^{J} P_j,
\label{eq:audit-partition}
\end{equation}
where $C_p$, $A_1,\ldots,A_S$, and $P_1,\ldots,P_J$ are pairwise disjoint.

The pruned strata may originate from view-based rewrites or proxy score bands. Their memberships and escalation order are fixed before any audit sample is drawn. All candidates are evaluated by the oracle, while the two families of assumed strata are audited as described below.

\subsection{Exact Finite-Population Bounds}
\label{sec:audit-bounds}

Consider a finite population containing $M$ tuples, of which an unknown number $T$ are positive. If a uniform sample of size $m$ is drawn without replacement and contains $k$ positives, then
\(
    X\sim\mathrm{Hypergeom}(M,T,m).
\)
For a one-sided error level $\delta$, define the exact lower and upper confidence bounds
\begin{align}
    L_\delta(M,m,k)
    &=
    \min\left\{
        \ell:
        \Pr_{T=\ell}[X\ge k]\ge\delta
    \right\},\\
    U_\delta(M,m,k)
    &=
    \max\left\{
        u:
        \Pr_{T=u}[X\le k]\ge\delta
    \right\}.
\end{align}
These bounds are obtained by inverting the hypergeometric distribution, analogously to the Clopper--Pearson construction~\cite{clopper1934use}, and satisfy
\begin{align}
    \Pr_T\!\left[L_\delta(M,m,X)\le T\right]
    &\ge 1-\delta,\\
    \Pr_T\!\left[T\le U_\delta(M,m,X)\right]
    &\ge 1-\delta
\end{align}
for every $T\in\{0,\ldots,M\}$. No distributional or asymptotic approximation is required.

\subsection{Joint Recall and Precision Audit}
\label{sec:audit-construction}

\subsubsection{Auditing assumed-positive strata.}
For each assumed-positive stratum $A_i$, the audit draws a uniform sample $\Sigma_i$ and evaluates it with the oracle. Let $m_i=|\Sigma_i|$ and let $k_i$ be the number of sampled positives. Sampled negatives are removed from the reported result, while sampled positives become oracle-verified.

The audit computes a lower bound $L_i = L_{\delta_i}(|A_i|,m_i,k_i)$ on the total number of positives in $A_i$. A stratum whose audit result does not meet the required precision floor is \emph{demoted} and evaluated completely. Let $\mathcal K\subseteq\{1,\ldots,S\}$ denote the retained strata. For each retained stratum, at least
\(
    \max\{0,L_i-k_i\}
\)
additional positives remain among its unsampled tuples. The total credited contribution is therefore
\begin{equation}
    \mathit{extra}^{+}
    =
    \sum_{i\in\mathcal K}
    \max\{0,L_i-k_i\}.
    \label{eq:extra-positive}
\end{equation}

\subsubsection{Auditing pruned strata.}
Rather than certifying every pruned stratum separately, \system\ pools them: $P^{(0)} = \bigcup_{j=1}^{J}P_j$. A uniform sample is drawn from this pool. Pooling avoids paying for a separate confidence bound for every source view while retaining stratum identities for correction and escalation.

At an escalation level $\ell$, let $P^{(\ell)}$ be the active pruned pool, let $m^{(\ell)}$ be the number of sampled tuples that fall in this pool, and let $k^{(\ell)}$ of them be positive. The exact upper bound $U^{(\ell)} = U_{\eta_\ell} \left( |P^{(\ell)}|, m^{(\ell)}, k^{(\ell)} \right)$ bounds the total number of positives in the active pool. Because the $k^{(\ell)}$ sampled positives are corrected and included in the result, the number of positives that may remain undetected is bounded by
\begin{equation}
    \mathit{miss}^{(\ell)}
    =
    \max\{0,U^{(\ell)}-k^{(\ell)}\}.
    \label{eq:miss-bound}
\end{equation}

\subsubsection{Recall certificate.}
Let $F$ denote the number of oracle-verified positives in the final result. This count includes positives found among candidates, audit samples, demoted strata, and escalated strata. If the audit terminates at level $\ell$, then a lower bound on recall is
\begin{equation}
    b_R
    =
    \frac{
        F+\mathit{extra}^{+}
    }{
        F+\mathit{extra}^{+}+\mathit{miss}^{(\ell)}
    }.
    \label{eq:bound}
\end{equation}
The numerator lower-bounds the number of reported true positives, while $\mathit{miss}^{(\ell)}$ upper-bounds the qualifying tuples still hidden in the active pruned pool.

\subsubsection{Precision certificate.}
Let $V = \bigcup_{i\in\mathcal K} (A_i\setminus\Sigma_i)$ be the unverified tuples retained from assumed-positive strata. Every reported tuple is either one of the $F$ oracle-verified positives or a member of $V$, so
\(
    |R_p|=F+|V|.
\)
The same lower bounds used in the recall certificate imply
\begin{equation}
    b_P
    =
    \frac{
        F+\mathit{extra}^{+}
    }{
        F+|V|
    }.
    \label{eq:precbound}
\end{equation}
Thus, precision is certified without additional oracle calls beyond those already used to audit the assumed-positive strata.

\subsubsection{Confidence allocation.}
The total error budget is divided before sampling: $\alpha_{\mathrm{pos}} + \alpha_{\mathrm{prune}} \le \alpha$. For the $S$ assumed-positive strata, we set
\(
    \delta_i
    =
    \displaystyle\frac{\alpha_{\mathrm{pos}}}{S}.
\)
For the $J+1$ possible active pruned pools in the escalation chain, we set
\(
    \eta_\ell
    =
    \displaystyle\frac{\alpha_{\mathrm{prune}}}{J+1},
    \; \ell=0,\ldots,J.
\)
Empty families are omitted, and their budget can be reassigned. Because all stratum and escalation-level bounds are simultaneously valid, subsequent demotion and data-dependent stopping do not invalidate the final certificate.

\subsection{Adaptive Sampling and Escalation}
\label{sec:audit-adaptive}

\subsubsection{Assumed-positive sample sizes.}
A stratum is retained only when its audit supports a precision floor $\rho\in(0,1)$. If its true precision were at most $\rho$, the probability that a without-replacement sample of size $m_i$ contained no negative tuple would be at most $\rho^{m_i}$. Consequently,
\begin{equation}
    m_i
    =
    \left\lceil
        \frac{\ln\delta_i}{\ln\rho}
    \right\rceil
    \label{eq:positive-sample-size}
\end{equation}
makes a clean sample sufficient evidence at level $1-\delta_i$. The exact hypergeometric lower bound $L_i$ is still used in the certificate. Because Equation~\ref{eq:positive-sample-size} does not depend on $|A_i|$, auditing a fixed number of assumed-positive strata need not scale with the corpus size.

\subsubsection{Pruned-pool sample size.}
After auditing the assumed-positive strata, let
\(
    G=F+\mathit{extra}^{+}
\)
be the current lower bound on reported true positives. To certify recall target $t$, the number of missed positives must not exceed
\(
    G\displaystyle\frac{1-t}{t}.
\)
Using a safety factor $\gamma\in(0,1)$, we define $A=\gamma G(1-t)/t$. For an active pool of size $M_{\mathrm{pool}}$ with $0<A<M_{\mathrm{pool}}$, a conservative initial sample size is
\begin{equation}
    m_P
    =
    \min\left\{
        M_{\mathrm{pool}},
        \left\lceil
            \frac{\ln\eta_0}
                 {\ln(1-A/M_{\mathrm{pool}})}
        \right\rceil
    \right\}.
    \label{eq:msize}
\end{equation}
This follows because, if the pool contained at least $A$ positives, the probability of observing no positive in a sample of size $m_P$ is at most
$(1-A/M_{\mathrm{pool}})^{m_P}$. The exact hypergeometric upper bound, rather than this relaxation, is used in the published certificate. If $A<1$, the pool is evaluated completely; if $A\ge M_{\mathrm{pool}}$, its worst-case number of misses is already compatible with the target.

\subsubsection{Escalation}
If the bound in Equation~\ref{eq:bound} remains below $t$, \system\ fully evaluates one pruned stratum, removes it from the active pool, and recomputes the certificate. The escalation order is fixed before audit sampling; for example, proxy score bands are ordered from higher to lower scores, with remaining ties resolved deterministically. Therefore, the possible active pools form a fixed nested sequence $P^{(0)} \supsetneq P^{(1)} \supsetneq \cdots \supsetneq P^{(J)} = \emptyset$. The stopping level may depend on the observed audit results, but simultaneous bounds are constructed for every level. Escalation can therefore continue until the target is certified or the active pool is empty, in which case recall is exact.

\subsection{Guarantees and Proofs}
\label{sec:claims}

\begin{theorem}[Recall certificate]
\label{thm:c2}
For any query answered through a fixed rewrite plan, regardless of whether
the predicted predicate relationships are correct,
\begin{equation}
    \Pr\!\left[
        \mathrm{recall}_p(R_p)\ge b_R
    \right]
    \ge 1-\alpha.
\end{equation}
The probability is over the fresh audit samples, conditional on the query
history and rewrite plan. The guarantee assumes only that the oracle defines
a fixed reference output for every predicate--tuple pair.
\end{theorem}

\begin{proposition}[Precision certificate]
\label{prop:prec}
The precision bound in Equation~\ref{eq:precbound} holds on the same
simultaneous-confidence event as the recall bound. Consequently,
\begin{equation}
    \Pr\!\left[
        \mathrm{recall}_p(R_p)\ge b_R
        \;\land\;
        \mathrm{precision}_p(R_p)\ge b_P
    \right]
    \ge 1-\alpha.
\end{equation}
\end{proposition}

\begin{proposition}[Graceful degradation]
\label{prop:c3}
With escalation enabled, \system\ can always reach the requested recall
target by evaluating additional pruned strata. In the worst case, the active
pruned pool becomes empty and execution approaches cold evaluation rather
than returning an uncertified result.
\end{proposition}

\begin{proposition}[Audit sample complexity]
\label{prop:c4}
Suppose the query selectivity is $s$, so that the number of reported true
positives is approximately $sN$. Under a clean pruned-pool sample,
Equation~\ref{eq:msize} gives $m_P = O\bigl(t\ln(1/\eta_0)/(\gamma(1-t)s)\bigr)$, which is independent of the corpus size $N$. This statement concerns the
initial audit sample; escalation cost can grow with the sizes of incorrectly
pruned strata.
\end{proposition}

\begin{proposition}[No cross-query error compounding]
\label{prop:c5}
Every audit evaluates fresh samples with the current oracle. Cached views
determine the proposed partition but do not supply evidence for the current
certificate. Therefore, a query derived through several generations of
cached views receives a new guarantee against the oracle, rather than
inheriting or multiplying the errors of earlier certificates.
\end{proposition}

\noindent\textit{Scope of pipeline composition.}
The per-query guarantees above do not imply that recall ratios can be
multiplied across successive semantic filters. End-to-end certification
requires propagating absolute missed-positive bounds and accounting for
false positives passed between stages. Developing such a pipeline-level
certificate is outside the scope of this work.

Equation~\ref{eq:msize} also identifies the regime in which certification
becomes expensive: the required sample grows as $t$ approaches one or as
selectivity $s$ decreases. The certificate remains valid in this regime,
but escalation may eliminate most of the savings over cold evaluation.

For any $X\subseteq D$, define
\(
    M(X)=|X\cap\ext(p)|.
\)
For each assumed-positive stratum $A_i$, let $\Sigma_i$ denote its uniform
audit sample. Let $\Sigma_P$ be the uniform sample drawn from the initial
pruned pool $P^{(0)}$. At escalation level $\ell$, define
\(
    \Sigma_P^{(\ell)}
    =
    \Sigma_P\cap P^{(\ell)},
    \;
    m^{(\ell)}
    =
    |\Sigma_P^{(\ell)}|,
    \;
    k^{(\ell)}
    =
    M(\Sigma_P^{(\ell)}).
\)

\begin{lemma}[Anatomy of the reported result]
\label{lem:anatomy}
At termination, $R_p$ is the disjoint union of the $F$ oracle-verified
positives and
\(
    V
    =
    \bigcup_{i\in\mathcal K}
    (A_i\setminus\Sigma_i).
\)
Hence $|R_p|=F+|V|$. If the audit stops at escalation level $L$, every
unreported positive lies in
$P^{(L)}\setminus\Sigma_P^{(L)}$.
\end{lemma}

\begin{proof}
Every tuple is a candidate, belongs to an audit sample, lies in a fully
evaluated demoted or escalated stratum, remains unverified in a retained
assumed-positive stratum, or remains unsampled in the active pruned pool.
Tuples in the first three categories are reported exactly according to the
oracle. Tuples in the fourth category constitute $V$, and only tuples in
the final category can contain an unreported positive.
\end{proof}

\begin{lemma}[Conditional uniformity]
\label{lem:cond}
Let $\Sigma$ be a uniform sample of size $m$ drawn without replacement
from a finite population $\Pi$. For any fixed $\Pi'\subseteq\Pi$,
conditional on $|\Sigma\cap\Pi'|=m'$, the restricted sample
$\Sigma\cap\Pi'$ is uniform among all size-$m'$ subsets of $\Pi'$.
\end{lemma}

\begin{proof}
For any $s'\subseteq\Pi'$ with $|s'|=m'$, the number of samples $\Sigma$
satisfying $\Sigma\cap\Pi'=s'$ is
\(
    \displaystyle\binom{|\Pi\setminus\Pi'|}{m-m'},
\)
which does not depend on the choice of $s'$. All size-$m'$ subsets are
therefore equally likely.
\end{proof}

\begin{lemma}[Fixed escalation chain]
\label{lem:chain}
Because the escalation order is determined by the plan before audit
sampling, the possible active pools
$P^{(0)},\ldots,P^{(J)}$ form a fixed nested sequence. Although the
termination level $L$ is data-dependent, $P^{(L)}$ is always selected
from this prespecified family.
\end{lemma}

\begin{proof}[Proof of Theorem~\ref{thm:c2} and Proposition~\ref{prop:prec}]
For each assumed-positive stratum, define
\(
    \mathcal E_i
    =
    \{L_i\le M(A_i)\}.
\)
For each escalation level, define
\(
    \mathcal G_\ell
    =
    \left\{
        M\!\left(
            P^{(\ell)}
            \setminus
            \Sigma_P^{(\ell)}
        \right)
        \le
        U^{(\ell)}-k^{(\ell)}
    \right\}.
\)
Exact hypergeometric inversion gives
$\Pr[\mathcal E_i]\ge1-\delta_i$.
By Lemmas~\ref{lem:cond} and~\ref{lem:chain}, the restriction
$\Sigma_P^{(\ell)}$ is a uniform sample from each prespecified active
pool, conditional on its size. Hence
$\Pr[\mathcal G_\ell]\ge1-\eta_\ell$.

The pruned-pool sample size may depend on the preceding assumed-positive
audit. However, the two families of samples use independent randomness
and disjoint tuple sets. Conditioning on the assumed-positive audit
therefore fixes the chosen sample size without changing the uniformity
of the subsequent pruned-pool sample.

Because
\(
    \sum_{i=1}^{S}\delta_i
    +
    \sum_{\ell=0}^{J}\eta_\ell
    \le\alpha,
\)
a union bound gives that
\(
    \Pr\!\left[
        \bigcap_{i=1}^{S}\mathcal E_i
        \cap
        \bigcap_{\ell=0}^{J}\mathcal G_\ell
    \right]
    \ge1-\alpha.
\)
This simultaneous event covers every stratum and every possible
termination level, so data-dependent demotion and escalation stopping
do not require additional assumptions.

On this event, every retained assumed-positive stratum satisfies
\(
    M(A_i\setminus\Sigma_i)
    =
    M(A_i)-k_i
    \ge
    \max\{0,L_i-k_i\}.
\)
Therefore,
\(
    M(V)\ge\mathit{extra}^{+}.
\)
Lemma~\ref{lem:anatomy} then gives
\(
    |R_p\cap\ext(p)|
    =
    F+M(V)
    \ge
    F+\mathit{extra}^{+},
\)
while $\mathcal G_L$ gives
\(
    |\ext(p)\setminus R_p|
    \le
    \mathit{miss}^{(L)}.
\)
Substituting these inequalities into the definitions of recall and
precision yields Equations~\ref{eq:bound} and~\ref{eq:precbound}.
Both therefore hold on the same event of probability at least
$1-\alpha$.
\end{proof}
\vspace{-9pt}
The argument does not use the predicted relationship labels or their
confidence scores. Incorrect predictions change which tuples enter each
stratum and may increase audit or escalation costs, but they do not
affect the validity of uniform finite-population sampling.

%% file: figures/fig_framework.tex
\begin{figure*}[t]
\centering
\begin{tikzpicture}[
  x=0.93cm, y=1cm, font=\small,
  ic/.style={draw=black!62, line width=0.55pt},
  lbl/.style={font=\fontsize{5}{5}\selectfont, align=center, text=black!70, anchor=north},
  stg/.style={font=\fontsize{5}{5}\selectfont\bfseries, text=black!88, anchor=south},
  ar/.style={-{Latex[length=4pt,width=3.2pt]}, draw=black!50, line width=0.55pt},
]
\definecolor{cand}{HTML}{2A78D6}
\definecolor{apos}{HTML}{EDA100}
\definecolor{prun}{HTML}{1BAF7A}
\definecolor{audc}{HTML}{EB6834}

\draw[ic, fill=black!5] (-0.40,0.34) -- (-0.40,-0.30) arc (180:360:0.40 and 0.12)
      -- (0.40,0.34) arc (0:180:0.40 and 0.12);
\draw[ic, fill=black!7] (0,0.34) ellipse (0.40 and 0.12);
\foreach \y in {-0.06,-0.22} \foreach \x in {-0.26,-0.13,0,0.13}
  {\fill[black!42] (\x,\y) rectangle ++(0.08,0.06);}
\node[stg] at (0,0.62) {1\ Predicate store};
\node[lbl, text width=2.5cm] at (0,-0.60) {cached NL views:\\text, bitmap, lineage};

\begin{scope}[shift={(3.05,0)}]
  \draw[ic, fill=black!4, rounded corners=1pt] (-0.70,0.06) rectangle (-0.14,0.44);
  \foreach \x in {-0.58,-0.44,-0.30} {\draw[ic] (\x,0.44) -- (\x,0.55);
                                      \draw[ic] (\x,0.06) -- (\x,-0.05);}
  \node[font=\fontsize{5}{5}\selectfont, text=black!72] at (-0.42,0.25) {NLI};
  \draw[ic, fill=black!4, rounded corners=2.4pt] (0.06,-0.34) rectangle (0.72,0.10);
  \node[font=\fontsize{5}{5}\selectfont, text=black!72] at (0.39,-0.12) {LLM};
  \draw[ar, dashed, dash pattern=on 1.3pt off 1.1pt] (-0.14,0.16) -- (0.04,-0.02);
  \node[stg] at (0,0.62) {2\ Two-tier reasoner};
  \node[lbl, text width=2.9cm] at (0,-0.60)
       {relation $+$ calibrated confidence;\\$O(\#\mathrm{views})$, not $O(N)$};
\end{scope}

\begin{scope}[shift={(6.20,0)}]
  \draw[ic, fill=black!4, rounded corners=1.6pt] (-0.62,-0.26) rectangle (0.62,0.42);
  \node[font=\fontsize{5}{5}\selectfont, text=black!72, align=center] at (0,0.08) {R1--R4\\bitmap algebra};
  \node[stg] at (0,0.62) {3\ Rewrite plan};
  \node[lbl, text width=2.6cm] at (0,-0.60)
       {Alg.~\ref{alg:plan}: one partition, fixed before any sampling};
\end{scope}

\draw[ar] (0.46,0.08) -- (2.28,0.08);
\draw[ar] (3.83,0.08) -- (5.52,0.08);
\draw[ar] (6.88,0.08) -- (7.72,0.08);

\begin{scope}[shift={(7.85,0)}]
  \fill[cand!80] (0,-0.24) rectangle (2.40,0.42);
  \fill[apos!80] (2.40,-0.24) rectangle (3.95,0.42);
  \fill[prun!80] (3.95,-0.24) rectangle (7.30,0.42);
  \foreach \x in {2.50,2.68,3.64,3.82}
      \fill[white] (\x,-0.16) rectangle ++(0.045,0.50);
  \foreach \x in {4.06,4.36,4.66,6.50,6.80,7.10}
      \fill[white] (\x,-0.16) rectangle ++(0.045,0.50);
  \draw[black!55, line width=0.5pt] (0,-0.24) rectangle (7.30,0.42);
  \draw[black!55, line width=0.4pt] (2.40,-0.24) -- (2.40,0.42);
  \draw[black!55, line width=0.4pt] (3.95,-0.24) -- (3.95,0.42);
  \node[font=\fontsize{5}{5}\selectfont, text=black!88] at (1.20,0.09) {candidates};
  \node[font=\fontsize{5}{5}\selectfont, text=black!88] at (3.17,0.09) {$+$};
  \node[font=\fontsize{5}{5}\selectfont, text=black!88] at (5.62,0.09) {pruned strata};
  \node[stg] at (3.65,0.62) {4\ The corpus, partitioned once};
  \node[lbl, text width=2.3cm] at (1.20,-0.38) {the oracle\\evaluates these};
  \node[lbl, text width=1.5cm] at (3.17,-0.38) {assumed $+$\\\emph{(precision)}};
  \node[lbl, text width=3.3cm] at (5.62,-0.38)
       {assumed $-$ \emph{(recall)}: one pooled sample, one exact bound};
\end{scope}

\begin{scope}[shift={(10.30,-1.68)}]
  \draw[ic, fill=white] (0,0.06) circle (0.16);
  \draw[ic, line width=0.85pt] (0.11,-0.05) -- (0.24,-0.18);
  \node[font=\fontsize{5}{5}\selectfont\bfseries, text=black!88, anchor=west] at (0.34,0)
       {5\ Audit \& escalate};
  \node[font=\fontsize{5}{5}\selectfont, text=black!70, anchor=west] at (2.62,0)
       {--- the white ticks are its own oracle calls};
\end{scope}

\begin{scope}[shift={(16.30,0)}]
  \draw[ar] (-1.06,0.08) -- (-0.44,0.08);
  \draw[ic, fill=audc!25] (0,0.20) circle (0.23);
  \node[font=\fontsize{5}{5}\selectfont, text=black!78] at (0,0.20) {$1{-}\alpha$};
  \draw[ic, fill=audc!25] (-0.14,0.02) -- (-0.08,-0.34) -- (0,-0.17)
        -- (0.08,-0.34) -- (0.14,0.02) -- cycle;
  \node[stg] at (0,0.62) {6\ Certificates};
  \node[lbl, text width=2.0cm] at (0,-0.60)
       {recall $\ge b$ \emph{and}\\precision $\ge c$, jointly};
\end{scope}
\end{tikzpicture}
\caption{\system's query path. Steps 1–3 require no tuple-level oracle evaluations. Tuple-level cost arises from the optional proxy pilot, candidate evaluation, audit, and escalation.}
\label{fig:framework}
\vspace{-15pt}
\end{figure*}
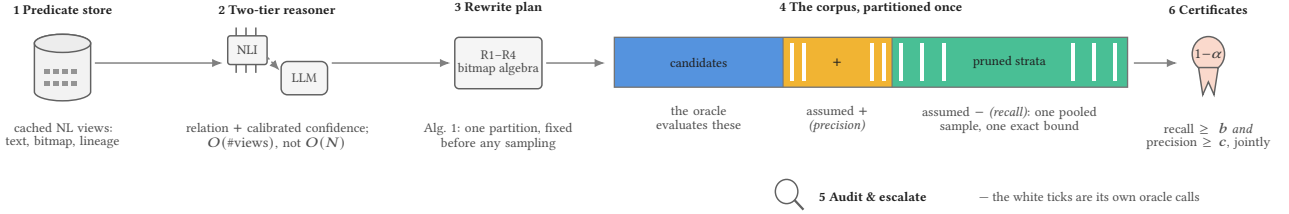

%% file: algoritms/buildplan.tex
\begin{algorithm}[t]
\small
\DontPrintSemicolon
\SetAlgoLined
\SetKwInOut{Input}{input}\SetKwInOut{Output}{output}
\Input{corpus size $N$; judgments $\{(V,\widehat{\mathit{rel}},c)\}$;
       thresholds $\tau$; cap $k_{\max}$}
\Output{a partition of $D$ into candidates and assumed strata}
$B,F,X \leftarrow$ views judged \rel{backward}/\rel{equiv}, \rel{forward}/\rel{equiv},
\rel{disjoint} with $c\ge\tau$\;
\lForEach{of the three sets}{keep its $k_{\max}$ highest-confidence views}
\ForEach{$V\in B$}{claim $R_V$ as an \textsc{AssumedPos} stratum\;}
\ForEach{$V\in F$}{claim $D_V\setminus R_V$ as \textsc{PrunedSuperset}\;}
\ForEach{$V\in X$}{claim $R_V$ as \textsc{PrunedDisjoint}\;}
\tcp{first claim wins; each stratum keeps its source view}
demote rows claimed both positive and pruned back to \textsc{Candidate}
\tcp*{conflicts are re-evaluated: always safe}
\Return candidates $\cup$ assumed-pos strata $\cup$ pruned strata\;
\caption{\textsc{BuildPlan} --- one partition, fixed before any audit sample is drawn.}
\label{alg:plan}
\end{algorithm}

%% file: experiment.tex
\section{Evaluation}
\input{tables/lotus_numbers}
\label{sec:eval}

We evaluate \system\ along five questions. \textbf{RQ1} asks whether cross-query reuse reduces oracle evaluations relative to cold execution, conventional caches, and a per-query proxy cascade. \textbf{RQ2} examines whether the result persists with a real LLM oracle and an existing semantic operator engine. \textbf{RQ3} studies scaling with corpus size and multi-label semantics. \textbf{RQ4} measures the validity, tightness, and cost of the joint recall--precision certificate. \textbf{RQ5} isolates the effects of the relation-classification tiers, rewrite rules, confidence thresholds, and workload composition. 
Label-based oracles support controlled studies at up to $804{,}414$ tuples; the real-LLM study uses a smaller corpus, since materializing a response matrix costs one LLM call per predicate--tuple pair. The former establishes scalability and the latter realism; no experiment combines both, and the real oracles are four open 8--27B models, not hosted frontier models.

\subsection{Experimental Setup}
\label{sec:eval-setup}

\subsubsection{Corpora and workloads}
We use 20~Newsgroups~\cite{lang1995newsweeder} ($7{,}951$ documents; 20 leaves in 7 groups), AG~News~\cite{zhang2015character} (up to $120{,}000$; 4 coarse classes) and RCV1-v2~\cite{lewis2004rcv1} (up to $804{,}414$; 103 topics in 4 levels, $3.2$ per tuple), which vary in taxonomy depth and label multiplicity, the two corpus properties most relevant to reuse. On a single-label corpus, predicates for distinct labels are extensionally disjoint; consequently, disjoint elimination applies to $99.2\%$ of queries on 20~Newsgroups. On RCV1-v2, a tuple may satisfy several topics, so we measure relations extensionally from the label matrix rather than infer them from topic names. The 20~Newsgroups and AG~News workloads contain leaf, group, union, and paraphrased predicates, with a default relatedness setting of $0.8$; RCV1-v2 sessions traverse its topic hierarchy. These generated workloads necessarily couple relatedness with predicate breadth and repetition. To reduce that dependence, the real-LLM experiment uses an analyst-written log of 47 session-ordered filters covering topical, rhetorical, and sentiment-oriented conditions. The log and its declared implications were fixed before any oracle calls, although they were written by the authors rather than harvested from a deployed system.

\subsubsection{Oracles}
The label oracle returns the corpus annotations at unit cost, enabling large-scale and controlled stress experiments. The real oracle is Llama-3.1-8B-Instruct~\cite{grattafiori2024llama3}, queried with a LOTUS-style \texttt{sem\_filter} prompt at temperature~0; Qwen2.5-14B~\cite{yang2024qwen25}, Mistral-Small-3.1-24B~\cite{mistral2025small31} and Gemma-3-27B~\cite{gemmateam2025gemma3} answer the same log with the same prompt (4-bit quantized; Gemma over the first 500 tuples). We materialize each distinct predicate--tuple response once and replay it under identical call accounting. A cache-free repetition of 200 responses reproduced all 200 for each model, supporting the deterministic-oracle assumption used by the audit. We separately evaluate integration with the released \texttt{lotus-ai} package. All simulated runs use seed~0; controlled stress configurations are repeated five times.

\subsubsection{Baselines and metrics}
We compare with \emph{cold} execution, an \emph{exact} normalized-text cache, an embedding cache that reuses a complete result at cosine similarity $\theta\in\{0.8,0.9\}$, and a per-query \emph{proxy} cascade. We also report \system\ without auditing and \system\ with ground-truth relations to isolate the cost of relation-classification errors. Unless stated otherwise, $t=0.9$, $\alpha=0.05$, and $\tau=(0.8,0.8,0.9,0.9)$. We measure oracle calls, macro precision and recall, relation-tier overhead, certified bounds, and bound violations. Token counts and latency are reported only for experiments that invoke a real LLM. Because each real-oracle call carries approximately the same prompt, call reductions also approximate token and monetary reductions in that setting.

\subsubsection{Implementation, hyperparameters, and environment}
\label{sec:eval-env}
All experiments ran on one Apple Mac Studio with an M3 Ultra (32 CPU cores, 24 performance and 8 efficiency; an 80-core GPU) and 512\,GB of unified memory, under macOS~26.5. \system\ is implemented in Python~3.12 with NumPy~2.5, SciPy~1.18 for the exact hypergeometric bounds, scikit-learn~1.9 for the calibration head and the proxy classifier, PyTorch~2.13, Transformers~5.15 and sentence-transformers~6.0; the NLI and embedding models run on PyTorch's Metal (MPS) device. Every LLM is served locally by Ollama with 4-bit (Q4\_K\_M) weights. Real-oracle answers are materialized once per (predicate, tuple) pair and replayed, so every method and setting sees identical answers. The LOTUS experiment instead calls the model live, through \texttt{lotus-ai}~1.2.4 and LiteLLM~1.98 in a separate environment, against Ollama~0.35.0 running four server instances, each with eight parallel slots and a $4{,}096$-token context, behind a round-robin proxy; each run checkpoints after every query pair and runs under a memory guard that stops it, and later resumes it, if the machine nears its memory limit. All randomness is seeded with $0$, and controlled stress configurations are repeated five times. Table~\ref{tab:hyperparams} lists the hyperparameters; code, workloads, response matrices, and per-query logs are in the artifact.

\begin{table}[t]
\caption{Hyperparameters. Defaults apply unless a study varies them.}
\label{tab:hyperparams}
\centering
\footnotesize
\setlength{\tabcolsep}{3pt}
\begin{tabular}{@{}p{1.45cm}p{6.6cm}@{}}
\toprule
\textbf{Component} & \textbf{Setting} \\
\midrule
Certificate & target recall $t=0.9$; failure probability $\alpha=0.05$, split by a union bound over pools and escalation levels; adaptive pruned-pool sample with allowance safety factor $\gamma=0.8$, from 20 to $4{,}000$ tuples per pool; pools under 10 tuples read in full; escalation largest stratum first \\
Precision & demotion floor $\rho=0.9$: an assumed-positive stratum is audited with $m_i=\lceil\ln\delta_i/\ln\rho\rceil$ samples and evaluated in full if it fails \\
Rewriting & $\tau_{\mathrm{eq}}=\tau_{\mathrm{fwd}}=0.8$, $\tau_{\mathrm{bwd}}=\tau_{\mathrm{dis}}=0.9$; $0.4$--$0.6$ in the threshold study and $0.5$ under LOTUS; at most 8 views per rule \\
Tier one & cross-encoder \texttt{nli-deberta-v3-xsmall}, scored in both directions; multinomial logistic head on the 6-dimensional scores, calibrated on 300 synthetic predicate pairs \\
Tier two & \texttt{qwen2.5:7b} at temperature 0, called when tier-one confidence is below $0.7$ or the prediction is overlap, at most 4 pairs per query; confidence is held-out class precision \\
LLM oracle & prompt \texttt{semfilter-v1} (answer yes or no), temperature 0, top-$k$ 1, seed 0, at most 3 output tokens; documents truncated to $1{,}200$ characters; \texttt{llama3.1:8b-instruct}, \texttt{qwen2.5:14b-instruct}, \texttt{mistral-small3.1:24b}, \texttt{gemma3:27b} \\
Baselines & embedding cache: \texttt{all-MiniLM-L6-v2}, cosine $\theta\in\{0.8,0.9\}$; proxy cascade: per-query logistic classifier on the same embeddings, 4 score bands, pilot $\{200,400,800\}$ $\times$ validation gate $\{0.6,0.7,0.8,0.9\}$ \\
Workloads & 120 queries, relatedness $0.8$ (swept over $[0,0.95]$); analyst log: 47 filters (46 distinct), 24 declared implications \\
Stress tests & entailment error $\{0,0.1,0.3\}$ $\times$ audit budget \{adaptive, 2\%, 5\%, 10\%, 20\%\} $\times$ $t\in\{0.8,0.9,0.95\}$; coverage: 540 cells $\times$ $20{,}000$ replicates; plan replay: $1{,}000$ re-audits per $(t,\alpha)$ \\
LOTUS & $N=1{,}000$, the whole log, two runs; \texttt{max\_batch\_size} 64; arms interleaved per query in alternating order \\
\bottomrule
\end{tabular}
\end{table}

\subsection{Main Results}

\subsubsection{End-to-End Effectiveness}
\label{sec:eval-e2e}

\begin{table}[t]
\footnotesize
\setlength{\tabcolsep}{2pt}
\caption{End-to-end comparison on 120-query workloads, with relatedness
$0.8$, $t=0.9$, and $\alpha=0.05$.}
\label{tab:main}
\begin{tabular}{l rrc rrc}
\toprule
 & \multicolumn{3}{c}{20NG ($N$=7{,}951)} & \multicolumn{3}{c}{AG News ($N$=8{,}000)} \\
\cmidrule(lr){2-4}\cmidrule(lr){5-7}
Method & calls & red. & P\,/\,R & calls & red. & P\,/\,R \\
\midrule
Cold (no reuse) & 954k & 1.00$\times$ & 1.000\,/\,1.000 & 960k & 1.00$\times$ & 1.000\,/\,1.000 \\
Exact cache & 660k & 1.45$\times$ & 1.000\,/\,1.000 & 216k & 4.44$\times$ & 1.000\,/\,1.000 \\
Embed cache $\theta$=0.9 & 644k & 1.48$\times$ & 0.996\,/\,0.996 & 216k & 4.44$\times$ & 1.000\,/\,1.000 \\
Embed cache $\theta$=0.8 & 429k & 2.22$\times$ & 0.900\,/\,0.915 & 136k & 7.06$\times$ & 0.925\,/\,0.942 \\
SemReuse w/o audit & 304k & 3.14$\times$ & 0.711\,/\,0.983 & 62k & 15.48$\times$ & 1.000\,/\,1.000 \\
SemReuse & 422k & 2.26$\times$ & 1.000\,/\,1.000 & 65k & 14.66$\times$ & 1.000\,/\,1.000 \\
SemReuse (oracle ent.) & 214k & 4.46$\times$ & 1.000\,/\,1.000 & 32k & 29.95$\times$ & 1.000\,/\,1.000 \\
\bottomrule
\end{tabular}
\end{table}

Table~\ref{tab:main} reports the principal comparison. On 20~Newsgroups, \system\ reduces oracle calls by $2.26\times$ while attaining measured $P=R=1.0$ and producing 100 recall and 100 precision certificates, none of which is violated. The embedding cache at $\theta=0.8$ achieves a similar $2.22\times$ reduction but lowers precision and recall to $0.900$ and $0.915$, respectively. Raising the threshold to $0.9$ restores accuracy but provides little beyond exact caching. This result does not contradict the use of embedding caches for near-duplicate requests: the mismatch arises because tuple-level predicate reuse is directional and may cover only part of the relation, whereas an embedding cache reuses an answer wholesale.

Relation-based rewriting lifts the exact cache's $1.45\times$ and $4.44\times$ to $2.26\times$ and $14.66\times$ ($4.46\times$ and $29.95\times$ with ground-truth relations): $1.57\times$ and $3.30\times$ beyond exact caching, so the gain is not attributable only to repeated predicate text; the larger gain on AG~News reflects its coarse label space and frequent recurrence.

The audit is necessary for reliable reuse. On 20~Newsgroups, removing it improves reduction from $2.26\times$ to $3.14\times$ but lowers precision to $0.711$. On AG~News, the unaudited output happens to remain exact because the relations are considerably easier; this cannot be established without oracle labels, which motivates certification.

\begin{figure}[t]
\centering
\input{figures/fig_decomposition}
\caption{Where the oracle calls go, as a share of cold evaluation.}
\label{fig:decomposition}
\vspace{-7pt}
\end{figure}
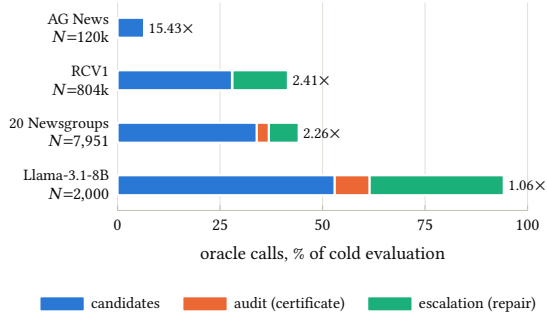

Figure~\ref{fig:decomposition} shows where the calls go. On 20~Newsgroups, candidates take $323{,}591$ calls, the audit $28{,}259$ ($3\%$ of cold) and escalation $69{,}657$ ($7\%$), which disappears with ground-truth relations; the relation tiers score only $6{,}806$ predicate pairs ($0.7\%$ of cold; $702$, or $0.07\%$, on AG~News). Across all experiments, the audit stays below $9\%$ of cold and often below $0.1\%$, whereas escalation ranges from zero on AG~News to a third or more of cold with real oracles: certification is cheap, but repairing an incorrect reuse decision can be expensive.

\subsubsection{Comparison and Composition with a Proxy Cascade}
\label{sec:eval-proxy}

The proxy baseline trains a logistic classifier from a per-query pilot sample, uses score-banded strata, and receives the same audit contract as \system. We sweep 12 pilot-size and activation-gate configurations and report its best result, obtained with a 400-tuple pilot and a $0.6$ activation gate. Reuse outperforms the proxy cascade by $1.40\times$ on this workload ($2.26\times$ vs.\ $1.62\times$; cold $954$k calls, cascade $590$k, \system\ $422$k) because the proxy pays pilot and audit costs for every query, whereas reuse benefits from an increasingly populated store. All reuse variants attain precision $1.0$; recall is $1.000$ for reuse alone and $0.992$--$0.994$ when composed with the proxy, with no certificate violations.

The two techniques compose safely but not profitably at the measured operating points. Composed, they reach $2.00\times$ ($2.84\times$ with exact relations, against $4.46\times$ for reuse alone); their joint audit produces no recall or precision violations, but each proxy band extends the escalation chain and divides the confidence budget over more active pools. Coarsening four proxy bands to one improves the composed ground-truth-relation result from $2.84\times$ to $3.79\times$, but provides little benefit with predicted relations. These results indicate that composition is most likely to help when the store is cold, and the proxy is strong.

Unlike SUPG~\cite{kang2020supg} and BARGAIN~\cite{zeighami2026bargain}, which certify a surrogate for the current query, the audit certifies past answers reused through fallible judgments, which may add as well as remove tuples; because it depends only on the final strata, not on whether they arose from an entailment judgment, a proxy score, or a hand-written rule, the two compose safely.

\subsection{LLM-Oracle Evaluation and Integration}
\label{sec:eval-realllm}

\begin{table}[t]
\caption{Experiments with a real LLM oracle. The first three rows use the 47-filter analyst-written workload; the last two evaluate LOTUS integration on the same log. ``Pred.'': predicates; ``Ora.'': oracle calls (per run for LOTUS); ``Red.'': reduction; ``Viol.'': certificate violations (for LOTUS, excluding shortfalls on bounds of exactly $1.0$, which only LOTUS's own nondeterminism can cause; Section~\ref{sec:eval-lotus}).}
\label{tab:real-oracle}
\centering
\footnotesize
\setlength{\tabcolsep}{2pt}
\begin{tabular}{p{2.5cm}rrrccr}
\toprule
\textbf{Execution} & \textbf{$N$} & \textbf{Pred.} &
\textbf{Ora.} & \textbf{Red.} & \textbf{$P/R$} &
\textbf{Viol.} \\
\midrule
Cold LLM oracle & $2{,}000$ & 47 & $94{,}000$ & $1.00\times$ &
$1.000/1.000$ & -- \\
\system\ w/o audit & $2{,}000$ & 47 & $47{,}748$ & $1.97\times$ &
$0.942/0.690$ & -- \\
\system & $2{,}000$ & 47 & $88{,}538$ & $1.06\times$ &
$1.000/0.994$ & $0$ \\
\midrule
LOTUS & $1{,}000$ & 47 & $47{,}000$ & $1.00\times$ & -- & -- \\
LOTUS + \system & $1{,}000$ & 47 & $45{,}373$ & $1.04\times$ & $0.998/0.992$ & $2$ \\
\bottomrule
\end{tabular}
\end{table}

\subsubsection{LLM oracle and analyst-written workload.}
The response matrix contains $92{,}000$ calls for the 46 distinct predicates, comprising $18.4$M prompt and $267$k output tokens. The model's responses, rather than corpus labels, define the reference semantics. Among 24 human-declared implications in the log, only the verbatim repeat holds exactly on the model-induced extensions; the median implication misses $9.3\%$ of the tuples selected by its antecedent. This finding shows that linguistic implication need not translate into exact containment under a fixed LLM oracle.

Approximate structure is nevertheless abundant: at a tolerated extension error of $\varepsilon=0.02$, $68\%$ of queries have an exploitable predecessor but only $11\%$ a containment one ($92\%$ and $49\%$ at $\varepsilon=0.10$). Yet at $t=0.9$ a containment missing $10\%$ of positives consumes the entire recall allowance before sampling uncertainty, and the containment rate falls from $28\%$ at $N=500$ to $15\%$ and $11\%$ at $N=1{,}000$ and $2{,}000$: relations that look exact on small samples need not stay so.

Table~\ref{tab:real-oracle} reports the end-to-end results. Rewriting reduces the candidate count from $94{,}000$ to $49{,}743$, but auditing requires $7{,}999$ calls and escalation restores another $30{,}796$. Consequently, \system\ achieves only a $1.06\times$ reduction, although its measured recall is $0.994$ and none of its certificates is violated. Without auditing, the same plans reach $1.97\times$ but reduce recall to $0.690$. A predictor trained from earlier audits reduces escalation by $22\%$ yet improves total cost by less than $1\%$, because most incorrectly pruned tuples are already recovered through escalation. This negative result points to partial stratum evaluation, rather than better cost prediction, as the more promising repair mechanism.

Likewise, the aggressive confidence thresholds that substantially improve the label-oracle result have negligible effect here. Together, these controls show that the binding constraint is the LLM oracle's inconsistency across related predicates, not simply a conservative relation classifier.

\begin{table}[t]
\caption{The analyst-written log on four LLM oracles. \emph{exact}~/~\emph{$\le$10\%}: declared implications (of 24) that hold exactly~/~miss by at most $10\%$, the recall allowance at $t{=}0.9$; \emph{miss}: median miss; \emph{esc.}: escalation, share of cold; \emph{perf.}: reduction with exact relations; \emph{unaudited}: the same plans without the audit. None of the $238$ recall and $238$ precision certificates is violated.}
\label{tab:multillm}
\centering
\footnotesize
\setlength{\tabcolsep}{3pt}
\begin{tabular}{l rrr rrr rr}
\toprule
 & \multicolumn{3}{c}{declared implications} & \multicolumn{3}{c}{\system} & \multicolumn{2}{c}{unaudited} \\
\cmidrule(lr){2-4}\cmidrule(lr){5-7}\cmidrule(lr){8-9}
Oracle & exact & $\le$10\% & miss & red. & esc. & perf. & red. & $R$ \\
\midrule
\multicolumn{9}{l}{\emph{$N{=}2{,}000$ tuples}} \\
\quad Llama-3.1-8B & 1 & 12 & 9.3\% & 1.06$\times$ & 33\% & 1.12$\times$ & 1.97$\times$ & 0.69 \\
\quad Qwen2.5-14B & 2 & 13 & 4.9\% & 1.06$\times$ & 36\% & 1.11$\times$ & 1.99$\times$ & 0.67 \\
\quad Mistral-Small-3.1-24B & 4 & 15 & 5.5\% & 1.05$\times$ & 39\% & 1.13$\times$ & 1.98$\times$ & 0.68 \\
\multicolumn{9}{l}{\emph{$N{=}500$ tuples}} \\
\quad Llama-3.1-8B & 5 & 13 & 7.4\% & 1.03$\times$ & 40\% & 1.04$\times$ & 1.97$\times$ & 0.69 \\
\quad Qwen2.5-14B & 8 & 12 & 8.7\% & 1.02$\times$ & 47\% & 1.03$\times$ & 1.99$\times$ & 0.66 \\
\quad Mistral-Small-3.1-24B & 12 & 16 & 0.9\% & 1.02$\times$ & 48\% & 1.03$\times$ & 1.99$\times$ & 0.68 \\
\quad Gemma-3-27B & 7 & 15 & 7.1\% & 1.03$\times$ & 42\% & 1.06$\times$ & 1.95$\times$ & 0.72 \\
\bottomrule
\end{tabular}
\end{table}

The finding is not one model's quirk (Table~\ref{tab:multillm}): no oracle honours the declared implications (at $N{=}2{,}000$ Qwen honours 2 of 24 exactly and Mistral 4; even at $N{=}500$ none honours more than half), and larger models miss less in the median but not in the tail (``medicine $\Rightarrow$ science'' misses by $62$--$66\%$ on Qwen and Mistral, $34\%$ on Llama). The economics do not move: \system\ saves $1.05$--$1.06\times$ on every full-size log, escalation takes $33$--$39\%$ of cold, and the unaudited plans lose $28$--$34\%$ of recall. Any two models agree on $92$--$97\%$ of tuples yet share only $49$--$70\%$ of positives, so a certificate can only be relative to the host's oracle.

Sweeping $N$ over the materialized answers separates audit cost from relation quality. Across an eight-fold increase in $N$, reduction with the NLI relation classifier remains between $1.03\times$ and $1.06\times$. With exact relations over the same model-induced extensions, it increases monotonically from $1.04\times$ to $1.19\times$. Candidate and escalation costs grow with $N$; only the certificate sample remains approximately size-independent. Hence, a flat audit cost yields growing end-to-end savings only when escalation remains small. The cost shares make the distinction explicit: from $N=500$ to $N=4{,}000$, candidates remain at $52.9$--$53.0\%$ of cold execution, audit changes from $4.3\%$ to $6.8\%$, and escalation remains dominant despite falling from $40.3\%$ to $34.2\%$. Oracle self-consistency, rather than corpus size alone, determines whether the constant audit cost becomes an asymptotic advantage.

\subsubsection{LOTUS integration.}
\label{sec:eval-lotus}

We run the released \texttt{lotus-ai} package, with its prompts, batching and parsing unchanged, over the whole analyst-written log and $N{=}1{,}000$ tuples, twice, in two arms per query: LOTUS as a user writes it, one \texttt{sem\_filter} over all tuples, and LOTUS behind \system, whose oracle is LOTUS's own \texttt{sem\_filter} over the tuples \system\ could not eliminate. The arms answer each query back to back in alternating order, and every count is LOTUS's own. Behind \system, LOTUS makes $\LotusCallPct\%$ fewer calls and spends $\LotusTokPct\%$ fewer tokens, the replay's $1.03\times$ at this size. Latency moves the other way: over the $\LotusCleanPairs$ pairs without a server retry, the \system\ arm takes $\LotusWallPct\%$ longer (paired per-query ratio: median $\LotusRatioMed$, interquartile range $\LotusRatioLo$--$\LotusRatioHi$), though its own computation is $\LotusOwn\%$ of its time: candidates, audit sample and escalations reach LOTUS as up to three sequential, smaller batches, which use a shared server worse than one batch of $1{,}000$; the audit sample depends only on the plan and could ride with the candidates.

The certificates hold up to the oracle's own noise: behind a batched server, LOTUS at temperature 0 does not reproduce itself, its two runs agreeing on $\LotusAgree\%$ of tuples. Of $\LotusCerts$ certificates, $\LotusShort$ fall short of the LOTUS-alone answers, by at most $\LotusMaxShort$ tuples; $\LotusShortExact$ carry a bound of exactly $1.0$, exact by construction since every answer behind it came from LOTUS in \system's own arm, so the shortfall is LOTUS answering a tuple differently in the other. The remaining $\LotusShortInexact$, one vague query's recall certificate in both runs, miss by one tuple where LOTUS's own runs differ in two; we count them as violations, within $\alpha$. Certifying across reruns thus requires pinning the serving stack, not only the temperature. LOTUS also maps an unparseable response to \emph{true}, which the precision certificate makes visible.

\subsection{Scaling with Corpus Size and Multi-Label Semantics}
\label{sec:eval-scaling}


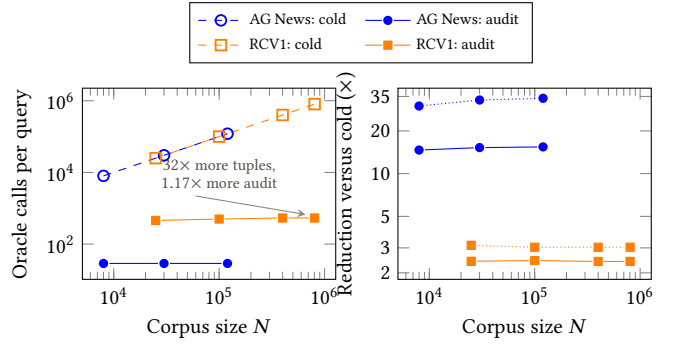
\begin{figure}[t]
\centering
\small
\begin{subfigure}[t]{0.485\linewidth}
\begin{tikzpicture}
\begin{axis}[
  width=1.2\linewidth,
  height=1.0\linewidth,
  xmode=log,
  ymode=log,
  log basis x=10,
  xlabel={Corpus size $N$},
  ylabel={Oracle calls per query},
  legend style={
    at={(1.1,1.45)},
    anchor=north,
    legend columns=2,
    font=\scriptsize,
    draw=black,
    /tikz/every even column/.append style={column sep=5pt}
  },
  legend cell align=left,
]
\addplot[
  blue,
  dashed,
  mark=o,
  sr open mark
] coordinates {
  (8000,8000)
  (30000,30000)
  (120000,120000)
};
\addlegendentry{AG News: cold}

\addplot[
  blue,
  mark=*,
  sr mark
] coordinates {
  (8000,28.8917)
  (30000,28.8917)
  (120000,28.8917)
};
\addlegendentry{AG News: audit}

\addplot[
  orange,
  dashed,
  mark=square,
  sr open mark
] coordinates {
  (25000,25000)
  (100000,100000)
  (400000,400000)
  (804414,804414)
};
\addlegendentry{RCV1: cold}

\addplot[
  orange,
  mark=square*,
  sr mark
] coordinates {
  (25000,456.8)
  (100000,499.575)
  (400000,533.683)
  (804414,534.367)
};
\addlegendentry{RCV1: audit}

\draw[
  ->,
  >=stealth,
  srMuted,
  line width=0.4pt
]
  (axis cs:6e4,2.2e3)
  --
  (axis cs:6.6e5,668.0);

\node[
  anchor=south,
  font=\scriptsize,
  text=srInkTwo,
  align=center
] at (axis cs:1e5,2.2e3)
  {$32\times$ more tuples,\\$1.17\times$ more audit};

\end{axis}
\end{tikzpicture}
\caption{Cold evaluation grows linearly with corpus size, whereas audit cost remains nearly constant.}
\label{fig:scaling-audit}
\end{subfigure}
\hfill
\begin{subfigure}[t]{0.485\linewidth}
\centering
\begin{tikzpicture}
\begin{axis}[
  width=1.2\linewidth,
  height=1.0\linewidth,
  xmode=log,
  ymode=log,
  log basis x=10,
  xlabel={Corpus size $N$},
  ylabel={Reduction versus cold ($\times$)},
  ytick={2,3,5,10,20,35},
  yticklabels={2,3,5,10,20,35},
  yminorticks=false,
  ymin=1.8,
  ymax=40,
  legend style={
    at={(0.5,-0.30)},
    anchor=north,
    legend columns=2,
    font=\scriptsize,
    draw=none,
    /tikz/every even column/.append style={column sep=5pt}
  },
  legend cell align=left,
]
\addplot[
  blue,
  solid,
  mark=*,
  sr mark
] coordinates {
  (8000,14.6639)
  (30000,15.2564)
  (120000,15.4264)
};

\addplot[
  blue,
  densely dotted,
  mark=*,
  sr mark
] coordinates {
  (8000,29.9476)
  (30000,33.0106)
  (120000,33.9578)
};

\addplot[
  orange,
  solid,
  mark=square*,
  sr mark
] coordinates {
  (25000,2.41135)
  (100000,2.44251)
  (400000,2.4058)
  (804414,2.40714)
};

\addplot[
  orange,
  densely dotted,
  mark=square*,
  sr mark
] coordinates {
  (25000,3.12561)
  (100000,3.02789)
  (400000,3.02699)
  (804414,3.02985)
};

\end{axis}
\end{tikzpicture}

\caption{End-to-end reduction grows only when relation quality keeps escalation small.}
\label{fig:scaling-reduction}
\end{subfigure}
\vspace{-5pt}
\caption{Scaling under label-based oracle.}
\label{fig:scaling}
\vspace{-7pt}
\end{figure}

Figure~\ref{fig:scaling} evaluates 120-query workloads as AG~News grows to $120{,}000$ tuples and RCV1-v2 grows to $804{,}414$. On AG~News, audit cost remains at 28.9 calls per query from $N=8{,}000$ to $N=120{,}000$, while the reduction improves from $14.66\times$ to $15.43\times$ with predicted relations and from $29.95\times$ to $33.96\times$ with exact relations. A fixed-fraction precision audit would instead grow from 62.6 to 552.4 calls per query, demonstrating the importance of size-independent sampling. No certificate is violated at any size ($0/107$ on AG~News and $0/84$ on RCV1-v2 at every $N$), and exact caching stays at $4.44\times$ and $1.85\times$.

RCV1-v2 is harder because predicates overlap. Audit cost rises only from 457 to 534 calls per query over a $32\times$ larger corpus, but reduction stays near $2.4\times$ with predicted relations and $3.0$--$3.2\times$ with exact ones, against $1.85\times$ for exact caching: disjoint elimination fires less often, hierarchical containments leave broad candidate sets (the superset topic \texttt{CCAT} covers $47\%$ of the corpus for a query on \texttt{C15}), and a wrong relation escalates a stratum that grows with $N$ ($2{,}936$ to $109{,}207$ escalation calls per query). Certification is nearly size-independent; repair can be linear. Taxonomy structure decides both how much can be pruned and what a mistake costs, which is why four coarse classes reach $15.4\times$ and 103 overlapping topics $2.4\times$.


\subsection{Certificate Validity and Robustness}
\label{sec:eval-audit}

We first corrupt relation predictions at rates $\{0,0.1,0.3\}$ and cross these levels with five audit budgets and three recall targets. The resulting 45 configurations comprise 225 workload runs at $N=4{,}000$ (60 queries and five repetitions per stress setting). No cell violates its nominal coverage. Across all runs with a deterministic oracle, three of $29{,}730$ recall certificates and none of $15{,}080$ precision certificates are violated, all within the $1-\alpha$ guarantee: two in the proxy sweep and one in the audit ablation below (the live LOTUS runs are discussed in Section~\ref{sec:eval-lotus}).

Reduction falls smoothly from $3.1\times$ to $1.7\times$ as relation-error probability rises from 0 to 0.3, while coverage remains intact: worse judgments increase cost rather than invalidate the certificate. The bounds are conservative in the tested workloads, however: the mean published recall bound is $0.957$, whereas mean realized recall is $1.0$, so most queries do not test the boundary.

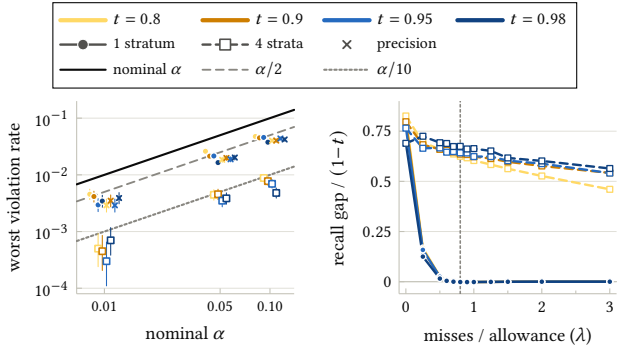
\begin{figure}[t]
\color{red}
\centering
\small
\ref{fig:coverage-legend}\par\medskip
\begin{subfigure}[t]{0.485\linewidth}
\centering
\input{figures/fig_coverage_violation}
\caption{Worst violation rate per $(t,\alpha)$ with 95\% intervals (markers offset by $t$): below nominal $\alpha$, at the union bound's share of it.}
\label{fig:coverage-violation}
\end{subfigure}
\hfill
\begin{subfigure}[t]{0.485\linewidth}
\centering
\input{figures/fig_coverage_tightness}
\caption{Realized minus certified recall at $\alpha{=}0.05$, in units of $1-t$, against the misses placed ($\lambda$).}
\label{fig:coverage-tightness}
\end{subfigure}
\vspace{-10pt}
\caption{Certificate coverage and tightness on controlled populations ($10.8$M audits; color is $t$).}
\label{fig:coverage}
\vspace{-3pt}
\end{figure}

We therefore also test the certificates where they are tight (Figure~\ref{fig:coverage}). Controlled populations of $5{,}000$ candidates with $1{,}000$ positives carry one pruned stratum of $14{,}500$ tuples or four ($8{,}000$ to $500$) plus a $500$-tuple assumed-positive stratum, with $\lambda G(1{-}t)/t$ missed positives placed in the pruned strata ($\lambda{=}1$ puts the unaudited plan exactly at the target). Each of 540 cells ($t\in\{0.8,0.9,0.95,0.98\}$, $\alpha\in\{0.01,0.05,0.1\}$, twelve $\lambda\in[0,3]$) resamples population and audit $20{,}000$ times through the shipped code. No cell exceeds its nominal $\alpha$; the worst sit at the union bound's share of $\alpha$---$\alpha/2$ for one stratum, just under $\alpha/10$ for four, below $\alpha/2$ for precision---and the hypergeometric primitive is exact to four decimals, so the distance to $\alpha$ is the union bound's split. With no misses the certificate sits $0.69$--$0.82\,(1{-}t)$ below realized recall, the price of the safety factor $\gamma{=}0.8$; the gap closes near $\lambda{=}0.8$. Re-auditing every rewritten plan of the Llama, Qwen and Mistral logs and of 20~Newsgroups $1{,}000$ times per $(t,\alpha)$ violates at most $0.16\%$ (recall) and $0.31\%$ (precision) of certificates in any cell.

\begin{table}[t]
\vspace{-5pt}
\caption{Audit-design ablation: call reduction with one decision changed at a time, all else fixed. Superscripts: measured precision where it drops below $1$ (uncertified); $^{*}$one recall violation in 34 certificates; $^{\ddagger}$the audit cannot certify and escalates. ``$+$~strata'': assumed-positive strata. ``a/q'' indicates the ratio of audit/query.}
\label{tab:auditdesign}
\vspace{-5pt}
\centering
\footnotesize
\setlength{\tabcolsep}{2pt}
\begin{tabular}{l rr rr rr}
\toprule
& \multicolumn{2}{c}{20NG} & \multicolumn{2}{c}{Llama log} & \multicolumn{2}{c}{AG News: a/q} \\
\cmidrule(lr){2-3}\cmidrule(lr){4-5}\cmidrule(lr){6-7}
Audit design & NLI & perf. & NLI & perf. & $N$=8k & $N$=120k \\
\midrule
\textbf{shipped}: pooled, adaptive, joint & 2.26$\times$ & 4.46$\times$ & 1.06$\times$ & 1.12$\times$ & 29 & 29 \\
\multicolumn{7}{l}{\emph{pooling}} \\
\quad one bound per stratum & 2.25$\times$ & 3.52$\times$ & 1.06$\times$ & 1.04$\times$ & 39 & 39 \\
\multicolumn{7}{l}{\emph{sample sizes}} \\
\quad fixed, $m{=}100$ & 1.49$\times$ & 1.54$\times$ & 1.03$\times$ & 1.05$\times$ & 38$^{\ddagger}$ & 38$^{\ddagger}$ \\
\quad fixed, $m{=}400$ & 1.68$\times$ & 2.06$\times$ & 1.04$\times$ & 1.10$\times$ & 153 & 153 \\
\quad fixed, $m{=}1{,}600$ & 2.03$\times$ & 2.88$\times$ & 1.03$\times$ & 1.02$\times$ & 613 & 613 \\
\quad fixed fraction, 2\% & 1.51$\times$ & 1.53$\times$ & 1.03$\times$ & 1.03$\times$ & 26$^{\ddagger}$ & 385 \\
\quad fixed fraction, 10\% & 1.91$\times$ & 2.94$\times$ & 1.03$\times$ & 1.04$\times$ & 128 & 1{,}925 \\
\multicolumn{7}{l}{\emph{what is certified}} \\
\quad recall only, $+$ strata sampled & 2.65$\times$$^{.72}$ & 4.46$\times$ & 1.18$\times$$^{.77}$ & 1.12$\times$ & 29 & 29 \\
\quad recall only, $+$ strata trusted & 1.95$\times$$^{.71}$ & 2.16$\times$ & 1.19$\times$$^{.74}$$^{*}$ & 1.12$\times$ & 4$^{\ddagger}$ & 4$^{\ddagger}$ \\
\bottomrule
\end{tabular}
\end{table}

Table~\ref{tab:auditdesign} changes one audit decision at a time; the Qwen and Mistral logs agree (extended version). \emph{Pooling} is up to $21\%$ ahead of one bound per pruned stratum when the audit is a large share of the bill ($4.46$ vs.\ $3.52\times$ on 20~Newsgroups with exact relations), even when that design gets its cost-minimizing allowance split, and falls behind only on the LLM logs with predicted relations, by at most $1.4\%$. \emph{Adaptive sizing} beats every fixed size and fraction in all twelve settings, untuned, whereas fixed fractions grow with $N$ ($128\to1{,}925$ audit calls per query at $10\%$ from $8$k to $120$k tuples, against a flat $29$). \emph{Joint certification} costs nothing when positive reuse is right and $9$--$17\%$ more calls when it is wrong, buying the precision repair ($0.72\to1.00$ on 20~Newsgroups, $0.77$--$0.80\to1.00$ on the LLM logs); not auditing assumed positives is usually costlier still ($1.95$ vs.\ $2.26\times$) and causes the ablation's one violation, a $0.45\%$ event under $4{,}000$ fresh seeds.

Flipping $5\%$ of oracle responses deterministically per predicate--tuple pair leaves the certificate valid relative to this modified oracle, and \system\ matches cold execution's label-space precision and recall ($0.635/0.951$ vs.\ $0.635/0.950$), while omitting the audit adds error. Reduction falls to $1.47\times$, as inconsistencies between related predicates trigger escalation. Certification preserves fidelity to what cold execution would return under the configured oracle; it does not turn a noisy oracle into ground truth---a limitation of the guarantee, not of this experiment.

\vspace{-3pt}
\subsection{Component and Workload Analysis}
\subsubsection{Relation-Classification Tiers}
\label{sec:eval-entailment}

The calibrated tier-one head improves accuracy by 31 points over direct NLI thresholding, reduces ECE from $0.52$ to $0.08$, and raises directional-relation $F_1$ from at most $0.08$ to $0.79$--$0.82$; growing the calibration set from 300 to 831 pairs moves accuracy only from $0.70$ to $0.73$, so the base NLI representation, not calibration size, is the limitation.

Tier two runs only when tier-one confidence is below $0.7$, or the prediction is overlap, for at most four pairs per query; its confidence is held-out class precision ($1.00$ equivalence, $0.84$ forward and $0.77$ backward containment, $0.46$ disjointness, $0.41$ overlap). Prompt structure matters: a five-way multiple-choice prompt answers 2 of 8 diagnostic cases, decomposing the decision into its defining binary questions 6 of 8.

On a fresh, balanced 300-pair set, the arbiter improves overlap recognition ($F_1$ $0.15\to0.36$) but fails on disjointness ($0.74\to0.06$; accuracy $0.740\to0.637$), partly because label-derived relations assume a closed-world corpus while the arbiter considers any document. Yet it lifts end-to-end reduction from $2.26\times$ to $2.50\times$ at unchanged $P=R=1.0$ with no violations: disjoint predictions have measured precision only $0.46$, below the rewrite threshold, so the system declines to act where the arbiter is weakest, and the audit protects the rest. Its 935 calls ($74{,}922$ prompt tokens) eliminate about $40{,}000$ tuple-level oracle calls.

\vspace{-3pt}
\subsubsection{Rewrite Rules and Confidence Thresholds}
\label{sec:eval-ablation}
The three rewrite rules provide complementary savings when exact relations are available: on 20~Newsgroups, exact reuse with the audit reaches $1.45\times$, adding R1 $2.71\times$, R2 $2.96\times$, and all rules $4.46\times$, each at $P=R=1.0$ with no violations. Under the default predicted-relation thresholds, disjoint elimination contributes little, because few disjoint predictions exceed $\tau_{\mathrm{dis}}=0.9$; lower thresholds recover much of this headroom. On 20~Newsgroups, setting all thresholds to $0.4$ raises reduction from $2.26\times$ to $3.47\times$ and cuts escalation from $69.7$k to $47.6$k calls; on AG~News the best setting is $0.6$, and lower ones create $5.9$k escalation calls. The audit thus turns threshold errors from a correctness risk into an efficiency trade-off, supporting an aggressive $0.4$--$0.6$ range rather than the conventional $0.8$--$0.9$, with a workload-dependent optimum; under the real LLM oracle, the same policy does not help, because approximate relations there disagree with the oracle's realized extensions.

\subsubsection{Workload Composition}
\label{sec:eval-overlap}

The relatedness generator changes the form of reuse, not only its amount: at setting $0.0$, $73\%$ of queries are exact-text repeats; at $0.95$, $90\%$ have a containment predecessor and mean selectivity falls from $0.145$ to $0.089$. Across these endpoints \system's reduction moves from $8.31\times$ to $2.68\times$ and the best embedding cache's from $3.75\times$ to $1.61\times$, with $P=R=1.0$ and no violations throughout: a transition from cheap exact reuse to partial containment-based reuse rather than a single-axis measure of reuse potential. \system\ beats the embedding cache by $1.7$--$2.2\times$ at every mixture, most where directional containment rather than near-duplicate matching is required.

Taken together, reuse pays under three conditions: predicates recur or relate to earlier ones; taxonomies are coarse and low-overlap, which keeps candidate and repair sets small; and the oracle answers related predicates consistently enough that escalation does not grow with the corpus. These explain the observed range, from $1.04\times$ under LOTUS to $17.84\times$ on AG~News, and caution against a single workload-independent speedup. Two practical rules follow: set relation thresholds aggressively, since the audit turns misclassification into cost, and publish precision with recall, since positive-union rewrites can add false positives at essentially no additional audit cost.

%% file: tables/lotus_numbers.tex
\providecommand{\LotusReps}{}\renewcommand{\LotusReps}{2}
\providecommand{\LotusPairs}{}\renewcommand{\LotusPairs}{94}
\providecommand{\LotusCleanPairs}{}\renewcommand{\LotusCleanPairs}{91}
\providecommand{\LotusRetried}{}\renewcommand{\LotusRetried}{3}
\providecommand{\LotusCallRed}{}\renewcommand{\LotusCallRed}{1.036}
\providecommand{\LotusTokRed}{}\renewcommand{\LotusTokRed}{1.036}
\providecommand{\LotusWallRed}{}\renewcommand{\LotusWallRed}{0.876}
\providecommand{\LotusCallPct}{}\renewcommand{\LotusCallPct}{3.5}
\providecommand{\LotusTokPct}{}\renewcommand{\LotusTokPct}{3.5}
\providecommand{\LotusWallPct}{}\renewcommand{\LotusWallPct}{14.2}
\providecommand{\LotusRatioMed}{}\renewcommand{\LotusRatioMed}{1.12}
\providecommand{\LotusRatioLo}{}\renewcommand{\LotusRatioLo}{1.04}
\providecommand{\LotusRatioHi}{}\renewcommand{\LotusRatioHi}{1.25}
\providecommand{\LotusOwn}{}\renewcommand{\LotusOwn}{0.15}
\providecommand{\LotusAgree}{}\renewcommand{\LotusAgree}{99.97}
\providecommand{\LotusAgreeMin}{}\renewcommand{\LotusAgreeMin}{99.7}
\providecommand{\LotusIdent}{}\renewcommand{\LotusIdent}{38}
\providecommand{\LotusCompared}{}\renewcommand{\LotusCompared}{47}
\providecommand{\LotusRecallCerts}{}\renewcommand{\LotusRecallCerts}{90}
\providecommand{\LotusPrecCerts}{}\renewcommand{\LotusPrecCerts}{90}
\providecommand{\LotusCerts}{}\renewcommand{\LotusCerts}{180}
\providecommand{\LotusShort}{}\renewcommand{\LotusShort}{28}
\providecommand{\LotusShortExact}{}\renewcommand{\LotusShortExact}{26}
\providecommand{\LotusShortInexact}{}\renewcommand{\LotusShortInexact}{2}
\providecommand{\LotusMaxShort}{}\renewcommand{\LotusMaxShort}{3}

%% file: figures/fig_decomposition.tex
\begin{tikzpicture}
\begin{axis}[sr base,
  width=0.64\columnwidth, height=2.75cm,
  xbar stacked, bar width=7.5pt, clip=false,
  xmin=0, xmax=100, xtick={0,25,50,75,100},
  xmajorgrids, ymajorgrids=false,
  xlabel={oracle calls, \% of cold evaluation},
  ytick={0,1,2,3}, yticklabels={{AG News\\$N{=}120$k},{RCV1\\$N{=}804$k},{20 Newsgroups\\$N{=}7{,}951$},{Llama-3.1-8B\\$N{=}2{,}000$}}, y dir=reverse,
  yticklabel style={align=right, font=\scriptsize, text=srInk},
  axis y line*=left, y axis line style={draw=none}, ytick style={draw=none},
  enlarge y limits=0.16,
  area legend,
  legend style={at={(0.42,-0.36)}, anchor=north, legend columns=3},
]
\addplot[fill=srBlue, draw=white, line width=0.8pt] coordinates {(6.45833,0) (27.9007,1) (33.9151,2) (52.9181,3)};
\addplot[fill=srOrange, draw=white, line width=0.8pt] coordinates {(0.0240764,0) (0.0664293,1) (2.96179,2) (8.50957,3)};
\addplot[fill=srAqua, draw=white, line width=0.8pt] coordinates {(0,0) (13.576,1) (7.30065,2) (32.7617,3)};
\legend{candidates, audit (certificate), escalation (repair)}
  \node[anchor=west, font=\scriptsize, text=srInk, inner sep=1.5pt] at (axis cs:6.482,0) {15.43$\times$};
  \node[anchor=west, font=\scriptsize, text=srInk, inner sep=1.5pt] at (axis cs:41.543,1) {2.41$\times$};
  \node[anchor=west, font=\scriptsize, text=srInk, inner sep=1.5pt] at (axis cs:44.178,2) {2.26$\times$};
  \node[anchor=west, font=\scriptsize, text=srInk, inner sep=1.5pt] at (axis cs:94.189,3) {1.06$\times$};
\end{axis}
\end{tikzpicture}

%% file: figures/fig_coverage_violation.tex
\begin{tikzpicture}[baseline=(current axis.south)]
\begin{axis}[sr base, width=0.70\linewidth, height=0.62\linewidth,
  xmode=log, ymode=log, xmin=0.0068, xmax=0.14, ymin=8e-5, ymax=0.2,
  xtick={0.01,0.05,0.1}, xticklabels={0.01,0.05,0.10},
  xminorticks=false, yminorticks=false,
  ytick={1e-4,1e-3,1e-2,1e-1},
  xlabel={nominal $\alpha$},
  ylabel={worst violation rate},
  legend to name=fig:coverage-legend, legend columns=4,
  legend style={draw=black, line width=0.4pt, inner sep=2pt,
                /tikz/every even column/.append style={column sep=5pt}},
]
\addplot[srInk, line width=0.8pt, domain=0.0068:0.14, samples=2, forget plot] {x};
\addplot[srMuted, densely dashed, line width=0.7pt, domain=0.0068:0.14, samples=2, forget plot] {x/2};
\addplot[srMuted, densely dotted, line width=0.8pt, domain=0.0068:0.14, samples=2, forget plot] {x/10};
\addlegendimage{srTa, line width=1.6pt}\addlegendentry{$t=0.8$}
\addlegendimage{srTb, line width=1.6pt}\addlegendentry{$t=0.9$}
\addlegendimage{srTc, line width=1.6pt}\addlegendentry{$t=0.95$}
\addlegendimage{srTd, line width=1.6pt}\addlegendentry{$t=0.98$}
\addlegendimage{srInkTwo, solid, mark=*, mark options={solid, draw=white, line width=0.4pt}}\addlegendentry{1 stratum}
\addlegendimage{srInkTwo, densely dashed, mark=square*, mark options={solid, fill=white, draw=srInkTwo, line width=0.6pt}}\addlegendentry{4 strata}
\addlegendimage{only marks, mark=x, srInkTwo, mark options={line width=0.7pt}}\addlegendentry{precision}
\addlegendimage{empty legend}\addlegendentry{}
\addlegendimage{srInk, line width=0.8pt}\addlegendentry{nominal $\alpha$}
\addlegendimage{srMuted, densely dashed, line width=0.7pt}\addlegendentry{$\alpha/2$}
\addlegendimage{srMuted, densely dotted, line width=0.8pt}\addlegendentry{$\alpha/10$}
\addlegendimage{empty legend}\addlegendentry{}
\addplot[srTa, only marks, mark=*, sr mark, mark size=1.2pt, forget plot, error bars/.cd, y dir=both, y explicit, error bar style={line width=0.4pt}, error mark=none] coordinates {(0.0081551,0.00455) += (0,0.0010335) -= (0,0.00088509) (0.040776,0.02605) += (0,0.0023026) -= (0,0.0021634) (0.081551,0.0475) += (0,0.0030394) -= (0,0.0029069)};
\addplot[srTb, only marks, mark=*, sr mark, mark size=1.2pt, forget plot, error bars/.cd, y dir=both, y explicit, error bar style={line width=0.4pt}, error mark=none] coordinates {(0.0086444,0.00415) += (0,0.00099199) -= (0,0.00084323) (0.043222,0.02135) += (0,0.0020993) -= (0,0.0019585) (0.086444,0.04495) += (0,0.0029634) -= (0,0.0028302)};
\addplot[srTc, only marks, mark=*, sr mark, mark size=1.2pt, forget plot, error bars/.cd, y dir=both, y explicit, error bar style={line width=0.4pt}, error mark=none] coordinates {(0.0091631,0.00295) += (0,0.00085366) -= (0,0.00070359) (0.045815,0.0215) += (0,0.0021061) -= (0,0.0019654) (0.091631,0.0457) += (0,0.002986) -= (0,0.002853)};
\addplot[srTd, only marks, mark=*, sr mark, mark size=1.2pt, forget plot, error bars/.cd, y dir=both, y explicit, error bar style={line width=0.4pt}, error mark=none] coordinates {(0.0097129,0.00345) += (0,0.00091419) -= (0,0.00076473) (0.048564,0.01645) += (0,0.0018598) -= (0,0.0017173) (0.097129,0.03765) += (0,0.0027311) -= (0,0.0025956)};
\addplot[srTa, only marks, mark=square*, mark options={solid, fill=white, draw=srTa, line width=0.6pt}, mark size=1.2pt, forget plot, error bars/.cd, y dir=both, y explicit, error bar style={line width=0.4pt}, error mark=none] coordinates {(0.0091631,0.0005) += (0,0.00041932) -= (0,0.00026021) (0.045815,0.00445) += (0,0.0010233) -= (0,0.00087481) (0.091631,0.0088) += (0,0.0013933) -= (0,0.0012476)};
\addplot[srTb, only marks, mark=square*, mark options={solid, fill=white, draw=srTb, line width=0.6pt}, mark size=1.2pt, forget plot, error bars/.cd, y dir=both, y explicit, error bar style={line width=0.4pt}, error mark=none] coordinates {(0.0097129,0.00045) += (0,0.00040407) -= (0,0.00024421) (0.048564,0.00455) += (0,0.0010335) -= (0,0.00088509) (0.097129,0.00775) += (0,0.0013146) -= (0,0.0011683)};
\addplot[srTc, only marks, mark=square*, mark options={solid, fill=white, draw=srTc, line width=0.6pt}, mark size=1.2pt, forget plot, error bars/.cd, y dir=both, y explicit, error bar style={line width=0.4pt}, error mark=none] coordinates {(0.010296,0.0003) += (0,0.00035286) -= (0,0.0001899) (0.051478,0.0035) += (0,0.00092) -= (0,0.00077059) (0.10296,0.00695) += (0,0.001251) -= (0,0.0011042)};
\addplot[srTd, only marks, mark=square*, mark options={solid, fill=white, draw=srTd, line width=0.6pt}, mark size=1.2pt, forget plot, error bars/.cd, y dir=both, y explicit, error bar style={line width=0.4pt}, error mark=none] coordinates {(0.010913,0.0007) += (0,0.0004742) -= (0,0.00031725) (0.054567,0.00385) += (0,0.00095952) -= (0,0.00081048) (0.10913,0.0048) += (0,0.0010585) -= (0,0.00091031)};
\addplot[srTa, only marks, mark=x, mark options={line width=0.7pt}, mark size=1.2pt, forget plot, error bars/.cd, y dir=both, y explicit, error bar style={line width=0.4pt}, error mark=none] coordinates {(0.010296,0.0029) += (0,0.00084733) -= (0,0.00069719) (0.051478,0.01845) += (0,0.0019616) -= (0,0.0018198) (0.10296,0.04015) += (0,0.0028133) -= (0,0.0026786)};
\addplot[srTb, only marks, mark=x, mark options={line width=0.7pt}, mark size=1.2pt, forget plot, error bars/.cd, y dir=both, y explicit, error bar style={line width=0.4pt}, error mark=none] coordinates {(0.010913,0.0035) += (0,0.00092) -= (0,0.00077059) (0.054567,0.01985) += (0,0.0020294) -= (0,0.0018881) (0.10913,0.04025) += (0,0.0028166) -= (0,0.0026819)};
\addplot[srTc, only marks, mark=x, mark options={line width=0.7pt}, mark size=1.2pt, forget plot, error bars/.cd, y dir=both, y explicit, error bar style={line width=0.4pt}, error mark=none] coordinates {(0.011568,0.0029) += (0,0.00084733) -= (0,0.00069719) (0.057841,0.0188) += (0,0.0019788) -= (0,0.0018371) (0.11568,0.0438) += (0,0.0029283) -= (0,0.0027947)};
\addplot[srTd, only marks, mark=x, mark options={line width=0.7pt}, mark size=1.2pt, forget plot, error bars/.cd, y dir=both, y explicit, error bar style={line width=0.4pt}, error mark=none] coordinates {(0.012262,0.0039) += (0,0.00096502) -= (0,0.00081603) (0.061311,0.0202) += (0,0.002046) -= (0,0.0019048) (0.12262,0.04205) += (0,0.0028739) -= (0,0.0027398)};
\end{axis}
\end{tikzpicture}

%% file: figures/fig_coverage_tightness.tex
\begin{tikzpicture}[baseline=(current axis.south)]
\begin{axis}[sr base, width=0.70\linewidth, height=0.62\linewidth,
  xmin=-0.1, xmax=3.1, ymin=-0.06, ymax=0.9,
  xtick={0,1,2,3}, minor xtick={0.5,1.5,2.5},
  ytick={0,0.25,0.5,0.75},
  xlabel={misses / allowance ($\lambda$)},
  ylabel={recall gap / $(1{-}t)$}]
\addplot[srAxis, line width=0.4pt, forget plot] coordinates {(-0.1,0) (3.1,0)};
\addplot[srMuted, densely dotted, line width=0.6pt, forget plot] coordinates {(0.8,-0.06) (0.8,0.9)};
\addplot[srTa, solid, mark=*, sr mark, mark size=1.1pt, forget plot] coordinates {(0,0.822891) (0.25,0.168138) (0.5,0.02668) (0.6,0.0108149) (0.7,0.00278427) (0.8,-0.000272034) (0.9,-0.00140828) (1,-0.00173567) (1.25,-0.000970255) (1.5,-0.000378522) (2,-4.21888e-05) (3,0)};
\addplot[srTb, solid, mark=*, sr mark, mark size=1.1pt, forget plot] coordinates {(0,0.800368) (0.25,0.159586) (0.5,0.0246164) (0.6,0.0100187) (0.7,0.00284182) (0.8,-0.00036072) (0.9,-0.00161454) (1,-0.00163954) (1.25,-0.001386) (1.5,-0.000567587) (2,-5.08163e-05) (3,0)};
\addplot[srTc, solid, mark=*, sr mark, mark size=1.1pt, forget plot] coordinates {(0,0.769231) (0.25,0.158435) (0.5,0.0247976) (0.6,0.00771467) (0.7,0.00182479) (0.8,-0.000793592) (0.9,-0.00169073) (1,-0.00154321) (1.25,-0.00119606) (1.5,-0.000764597) (2,-5.65611e-05) (3,0)};
\addplot[srTd, solid, mark=*, sr mark, mark size=1.1pt, forget plot] coordinates {(0,0.690335) (0.25,0.124923) (0.5,0.0154373) (0.6,0.00425856) (0.7,0) (0.8,-0.00159674) (0.9,-0.00193751) (1,-0.00171134) (1.25,-0.00106693) (1.5,-0.00040653) (2,0) (3,0)};
\addplot[srTa, densely dashed, mark=square*, mark options={solid, fill=white, draw=srTa, line width=0.5pt}, mark size=1.1pt, forget plot] coordinates {(0,0.825793) (0.25,0.694611) (0.5,0.659508) (0.6,0.647835) (0.7,0.636298) (0.8,0.625054) (0.9,0.617339) (1,0.603813) (1.25,0.583676) (1.5,0.56301) (2,0.526733) (3,0.460329)};
\addplot[srTb, densely dashed, mark=square*, mark options={solid, fill=white, draw=srTb, line width=0.5pt}, mark size=1.1pt, forget plot] coordinates {(0,0.796985) (0.25,0.682783) (0.5,0.659623) (0.6,0.655022) (0.7,0.646093) (0.8,0.635528) (0.9,0.633661) (1,0.622787) (1.25,0.614136) (1.5,0.59755) (2,0.576725) (3,0.542114)};
\addplot[srTc, densely dashed, mark=square*, mark options={solid, fill=white, draw=srTc, line width=0.5pt}, mark size=1.1pt, forget plot] coordinates {(0,0.765357) (0.25,0.664847) (0.5,0.666713) (0.6,0.645428) (0.7,0.648581) (0.8,0.642174) (0.9,0.64155) (1,0.623614) (1.25,0.623518) (1.5,0.603066) (2,0.58818) (3,0.541271)};
\addplot[srTd, densely dashed, mark=square*, mark options={solid, fill=white, draw=srTd, line width=0.5pt}, mark size=1.1pt, forget plot] coordinates {(0,0.689529) (0.25,0.724838) (0.5,0.692667) (0.6,0.688765) (0.7,0.673591) (0.8,0.674367) (0.9,0.661407) (1,0.663212) (1.25,0.651679) (1.5,0.616737) (2,0.601837) (3,0.564158)};
\end{axis}
\end{tikzpicture}

%% file: relatedwork.tex
\section{Related Work}
\label{sec:related}

\noindent\textit{Semantic operator systems.}
Systems such as LOTUS~\cite{patel2025lotus}, Palimpz-\\est~\cite{liu2025palimpzest}, DocETL~\cite{shankar2025docetl}, UQE~\cite{dai2024uqe}, ThalamusDB~\cite{jo2024thalamusdb}, CAESURA~\cite{urban2024caesura}, and Cortex AISQL~\cite{liskowski2025aisql} provide declarative interfaces for applying LLM-based operators to unstructured data, while TAG~\cite{biswal2025tag} considers the broader interaction between language models and data systems. Their optimizations primarily reduce the cost of executing an individual query or workflow. Examples include proxy cascades in LOTUS, cost-based plan search in Abacus~\cite{russo2025abacus}, learned optimization in Larch~\cite{zhao2026larch}, and agentic rewriting in DocETL. FrugalGPT~\cite{chen2023frugalgpt} similarly reduces inference cost through model cascades. \system\ is complementary to these techniques: it reuses results from previously evaluated predicates to reduce the tuple set presented to the host optimizer, which may itself employ proxies, cascades, batching, or other intra-query optimizations.

\vspace{3pt}
\noindent\textit{Proxy-based acceleration and statistical guarantees.}
Learned filters and inexpensive surrogate models have long been used to reduce calls to expensive predicates. Probabilistic predicates~\cite{lu2018probabilistic} and NoScope~\cite{kang2017noscope} use lightweight models to identify promising inputs, while SUPG~\cite{kang2020supg} provides recall and precision guarantees for proxy-assisted selection through importance sampling. BARGAIN~\cite{zeighami2026bargain} further develops adaptive techniques with statistical guarantees, and online aggregation~\cite{hellerstein1997online} established the broader principle of reporting progressively refined confidence bounds. These methods primarily control the error introduced by a surrogate for the current predicate. In contrast, \system\ audits decisions derived from previously computed results and fallible semantic-relation predictions. It treats all exclusions as provisional, applies exact finite-population bounds across the resulting strata, and jointly certifies recall and precision. Consequently, its guarantees do not require either the relation classifier or an optional proxy to be correct.

\vspace{3pt}
\noindent\textit{Cross-query reuse and semantic caching.}
Classical database systems reuse computation through multiple-query optimization~\cite{sellis1988mqo}, query answering using materialized views~\cite{levy1995answering,halevy2001aquv,chaudhuri1995optimizing,goldstein2001optimizing,pottinger2001minicon}, and semantic caching~\cite{dar1996semantic}. These approaches use formal query semantics to establish reuse conditions such as equivalence or containment. Because natural-language filters do not expose such relationships, \system\ predicts them and statistically audits the resulting reuse decisions. LLM-oriented semantic caches~\cite{bang2023gptcache,liu2025semanticcaching,wang2025categoryaware} instead reuse complete responses for similar requests, but similarity does not establish containment between the selected tuple sets or certify recall and precision. Eager-extraction systems~\cite{sun2025quest,lin2024zendb,arora2023evaporate} materialize an anticipated schema, whereas \system\ lazily reuses results from previously encountered predicates. Its relation module builds on sentence embeddings~\cite{reimers2019sbert} and NLI models~\cite{he2021deberta,he2023debertav3,bowman2015snli,williams2018multinli}. Because NLI scores may not represent predicate-extension relationships and can be miscalibrated~\cite{guo2017calibration}, they guide reuse planning only; the audit remains responsible for correctness.

%% file: conclusion.tex
\section{Conclusion}
\label{sec:conclusion}

We present \system, a framework for reusing earlier semantic-filter results despite uncertain relationships between natural-language predicates. \system uses a two-tier relation classifier to construct reuse plans and a pooled stratified audit to jointly certify recall and precision. The guarantee is defined relative to the host oracle and does not assume that predicted relationships are correct; prediction errors affect cost rather than certificate validity. Experiments across multiple corpora, workloads, oracle settings, and LOTUS show that \system\ can substantially reduce tuple-level oracle calls. They also identify oracle inconsistency across related predicates as a key practical constraint. The current scope is unary filters with per-query certification, and savings depend on recurring predicates and useful cached results. 
%
Future work will investigate more cost-effective escalation strategies, approaches for sharing audit evidence across related queries, and extensions of the framework to semantic joins over tuple pairs.